\documentclass[11pt,showpacs,twocolumn,pra,superscriptaddress,notitlepage]{revtex4-2}

\usepackage{qcircuit}
\usepackage[dvips]{graphicx}
\usepackage{amsmath,amssymb,amsthm,mathrsfs,amsfonts,dsfont}
\usepackage{subfigure, epsfig}
\usepackage{braket}
\usepackage{bm}
\usepackage{enumerate}
\usepackage{algorithm}
\usepackage{algpseudocode}
\usepackage{diagbox}
\usepackage{physics}
\usepackage{color}
\usepackage{geometry}
\usepackage{multirow}
\usepackage{comment}
\usepackage{tikz}

\usepackage[]{qcircuit}

\usetikzlibrary{arrows}
\usetikzlibrary{shapes,fadings,snakes}
\usetikzlibrary{decorations.pathmorphing,patterns}
\usetikzlibrary{calc}
\usetikzlibrary{positioning}
\usepackage[colorlinks = true]{hyperref}

\usepackage[font=small,labelfont=bf,justification=centerlast]{caption}

\graphicspath{{./figure/}}

\newtheorem{theorem}{Theorem}
\newtheorem{task}{Task}
\newtheorem{observation}{Observation}

\newtheorem{lemma}{Lemma}

\newtheorem{prop}{Proposition}
\newtheorem{definition}{Definition}

\newcommand{\mc}{\mathcal}
\newcommand{\mb}{\mathbf}
\newcommand{\mbb}{\mathbb}
\newcommand{\id}{\mathbb{I}}

\newcommand{\comments}[1]{}

\newcommand{\bk}[1]{\mbox{$\left\langle #1 \right\rangle$}}

\usepackage{orcidlink}
\usepackage{hyperref} 
\hypersetup{colorlinks=true,linkcolor=blue}

\begin{document}

\let\oldacl\addcontentsline
\renewcommand{\addcontentsline}[3]{}

\title{Minimal Clifford Shadow Estimation by Mutually Unbiased Bases}

\author{Qingyue Zhang}
\affiliation{Key Laboratory for Information Science of Electromagnetic Waves (Ministry of Education), Fudan University, Shanghai 200433, China}
\author{Qing Liu~\orcidlink{0000-0002-1576-1975}}
\affiliation{Key Laboratory for Information Science of Electromagnetic Waves (Ministry of Education), Fudan University, Shanghai 200433, China}
\author{You Zhou~\orcidlink{0000-0003-0886-077X}}
\email{you\_zhou@fudan.edu.cn}
\affiliation{Key Laboratory for Information Science of Electromagnetic Waves (Ministry of Education), Fudan University, Shanghai 200433, China}
\affiliation{Hefei National Laboratory, Hefei 230088, China}

\begin{abstract}
Predicting properties of large-scale quantum systems is crucial for the development of quantum science and technology. Shadow estimation is an efficient method for this task based on randomized measurements, where many-qubit random Clifford circuits are used for estimating global properties like quantum fidelity. Here we introduce the minimal Clifford measurement (MCM) to reduce the number of possible random circuits to the minimum, while keeping the effective post-processing channel in shadow estimation. In particular, we show that MCM requires $2^n+1$ distinct Clifford circuits, and it can be realized by Mutually Unbiased Bases (MUB), with $n$ as the total qubit number. By applying the Z-Tableau formalism, this ensemble of circuits can be synthesized to the $\mathrm{-S-CZ-H-}$ structure, which can be decomposed to $2n-1$ \emph{fixed} circuit modules, and the total circuit depth is at most $n+1$. Compared to the original Clifford measurements, our MCM reduces the circuit complexity and the compilation costs. In addition, we find the sampling advantage of MCM on estimating off-diagonal operators, and extend this observation to the biased-MCM scheme to enhance the sampling improvement further.

\end{abstract}

\maketitle

\twocolumngrid
\section{Introduction}
Learning quantum systems is of both fundamental and practical interests for quantum physics and quantum information processing \cite{gebhart2023learning,anshu2023survey}. 
With the increase on the qubit-number in various platforms, from quantum networks to quantum simulators and computers \cite{PRXQSimulator,PRXQComputer}, it is crucial to develop efficient tools to benchmark them \cite{Eisert2020certification,Kliesch2021Certification}, from characterizing quantum noises \cite{emerson2005scalable} to measuring interesting properties, such as entanglement entropy \cite{Brydges2019Probing} and out-of-time-ordered correlator \cite{mi2021information}.
Shadow tomography \cite{huang2022learning} is a recently proposed framework to predict many properties of quantum objects, like quantum states and channels with randomized measurements \cite{elben2023randomized}. Compared to traditional quantum tomography aiming to reconstruct quantum objects \cite{haah2017sample,flammia2012quantum}, the key point of shadow estimation is to build a few classical snapshots of the original quantum object \cite{huang2020predicting}, which are utilized to estimate many observables in a multiplex way.

\begin{figure}[htbp]
   \centering \includegraphics[width=0.44\textwidth]{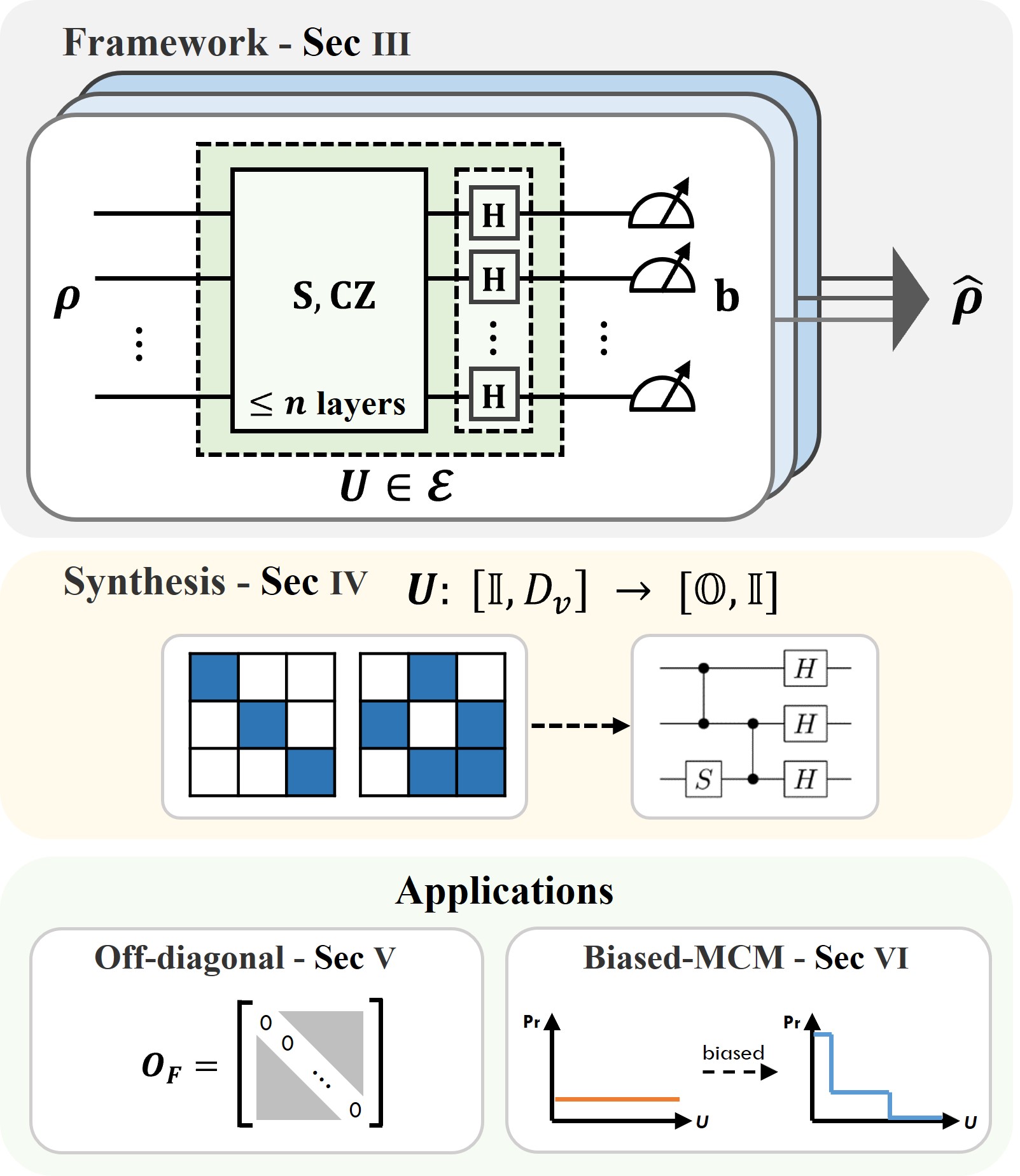}
    \caption{The outline of MCM shadow estimation.}
    \label{fig:main}
\end{figure}

The performance of shadow estimation depends on the applied random unitary evolution on the unknown state and the observables to be predicted \cite{huang2020predicting,hu2023classical,Hu2022Hamiltonian}. There are two primary random unitary ensembles \cite{huang2020predicting}. Pauli measurements enabled by independent single-qubit random unitary rotation is efficient to predict local observables with a constant non-trivial support. On the other hand, Clifford measurements using random Clifford evolution on all $n$ qubits, are efficient for global low-rank observables, such as the fidelity to some entangled states. The Pauli measurement and its variants are more feasible to realize \cite{huang2021efficient,hadfield2022measurements,wu2023overlapped,van2022hardware,zhou2023performance,zhou2022hybrid}, with a few experimental demonstrations \cite{zhang2021experimental,Stricker2022experimental,an2023efficient}.  

In contrast, the development and realization of the Clifford measurement is more challenging, which is mainly due to the complex structure of Clifford group with an astronomical number about $O(4^{n^2})$ of elements \cite{ozols2008clifford,bravyi20226}. In addition, it generally needs $O(n)$-depth quantum circuit with considerable sampling \cite{koenig2014efficiently,van2021simple} and compiling efforts \cite{maslov2022depth}, which further hinders the applications on the near-term quantum platform. There are positive progresses, for example, applying random local quantum gates sequentially to approximate full Clifford measurements \cite{hu2023classical,bertoni2022shallow} or using emergent-design with the assistance of large number of ancilla qubits \cite{tran2023measuring,mcginley2022shadow}. However, the performance is generally guaranteed in the average-case \cite{bertoni2022shallow}, and the classical post-processing is in an approximate and empirical manner \cite{hu2023classical}.

To accelerate the application of Clifford measurement, in this work, we propose the Minimal Clifford measurement (MCM) framework for shadow estimation, to avoid the exhaustion of the full Clifford group. More specifically, we reduce the number of sampled Clifford circuits to its minimum, and meanwhile the post-processing maintains the simple form as the original one. In particular, we prove that such minimal set should contain $2^n+1$ elements, which can be realized by Mutually Unbiased Bases (MUB) (Sec.~\ref{sec:MCM}). We further give a general routine to synthesize these Clifford circuits with the help of Tableau formalism, and the final circuit structure is in the $\mathrm{-S-CZ-H-}$ form with the depth at most $n+1$ (actually can be further reduced to $n/2$ \cite{maslov2022depth}) over the all-to-all architecture. Even though the circuit depth still scales with $n$, 
$\mathrm{S}$ and $\mathrm{CZ}$ gates are diagonal in the computational basis which may be realized simultaneously \cite{bremner2016average}, potentially with the help of Ising-type Hamiltonian \cite{nakata2017efficient}. Most importantly, the $\mathrm{-S-CZ-}$ part can be further decomposed to $2n-1$ fixed modules with an explicit sub-circuit structure. Such module decomposition enables the sampling and synthesis of random Clifford circuits on the module-level in a unified manner, which is more feasible to benchmark and improve experimentally \cite{liu2023group} (Sec.~\ref{sec:circuit}). We give a thorough performance analysis of MCM shadow analytically and numerically. In particular, we find our approach shows some advantage for estimating off-diagonal observables, and relate the variance of the estimation to the \emph{coherence} of the observable (Sec.~\ref{sec:performance}). Furthermore, we develop the biased-MCM protocol as a mimic of biased-Pauli measurement \cite{huang2021efficient,hadfield2022measurements} in the Clifford scenario and demonstrate its advantages. In particular, we show that the optimal variance of biased-MCM is directly quantified by the stabilizerness-norm \cite{campbell2011catalysis} of the observable, a useful measure of magic \cite{Howard2017Resource,leone2022stabilizer}, which may lead to potential applications in fidelity estimation (Sec.~\ref{sec:biased-MCM}).                                       

\section{Preliminaries for shadow estimation and Clifford measurements}\label{sec:review}
Here we first give a brief review \cite{huang2020predicting} on the paradigm of shadow estimation and Clifford group \cite{aaronson2004improved}. For the quantum experiment, an unknown $n$-qubit quantum state $\rho \in \mc{H}_d$ with $d=2^n$ is evolved under an unitary $U$, which is randomly selected from some ensemble $\mc{E}$, to $\rho \mapsto U\rho U^{\dag}$. After that, it is measured in the computational basis to get the result $\mb{b}\in \{0,1\}^n$.  According to Born's rule, the corresponding probability is $\Pr(\mb{b}|U)=\langle \mb{b}| U\rho U^{\dag} |\mb{b}\rangle$. Note that both $U$ and $\mb{b}$ are random variables and the whole process is indeed a random one. 

For the classical post-processing, one `prepares' $U^{\dag} \ket{\mb{b}}\bra{\mb{b}}U$ on the classical computer, and the effective process can be written as a quantum channel.
For a fixed $U$, that is, by taking the expectation first on $\mb{b}$, we denote the effective channel as 
\begin{equation}\label{MUchannel}
\begin{aligned}
\mathcal{M}(\rho|U):&=\mbb{E}_{\mb{b}}\ U^{\dag} \ket{\mb{b}} \bra{\mb{b}}U\\
&=\sum_{ \mb{b}}\tr[\rho U^{\dag} \ket{\mb{b}} \bra{\mb{b}}U] U^{\dag} \ket{\mb{b}} \bra{\mb{b}}U,
\end{aligned}
\end{equation}
where the conditional probability $\Pr(\mb{b}|U)$ is in the trace form for later convenience. Thus the whole channel is obtained by further taking the expectation on $U$,
\begin{equation}\label{Mchannel}
\mathcal{M}_\mc{E}(\rho):= \mbb{E}_{U\in \mc{E}}\ {\mathcal{M}(\rho|U)}.
\end{equation}
If the measurement is tomographically (over-) complete, i.e., one takes sufficient distinct $U$ for evolution, the whole information of $\rho$ should preserve. In other words, the channel $\mathcal{M}_\mc{E}$ can be reversed mathematically, and one can construct the classical snapshot as 
\begin{equation}\label{eq:shadowbasic}
\hat{\rho}=\mathcal{M}_\mc{E}^{-1}\ (U^{\dag} \ket{\mb{b}} \bra{\mb{b}}U).
\end{equation}
It is direct to check that $\mbb{E} \hat{\rho}=\rho$ by Eq.~\eqref{MUchannel} and \eqref{Mchannel}, and $\hat{O}:=\tr[\hat{\rho}O]$ is an unbiased estimator of $\bar{O}:=\tr[\rho O]$. For applications, one can construct the shadow set containing a few of these independent snapshots $\{\hat{\rho}_{i}\}$ to predict many properties $\{O_j\}$.

There are two prominent unitary ensembles, that is, $n$-qubit random Clifford ensemble $\mc{E}_{\mathrm{Cl}}$ and the tensor-product of random single-qubit Clifford gate ensemble $\mc{E}_{\mathrm{Pauli}}$. The Pauli measurement by the ensemble $\mc{E}_{\mathrm{Pauli}}$ can be realized efficiently in experiments, but it works poorly on estimating global observables, such as the fidelity to some many-qubit entangled states. Clifford measurement $\mc{E}_{\mathrm{Cl}}$ is good at this task, with the effective channel and its inverse for Clifford measurement being
\begin{equation}\label{Ch&Inv}
\begin{aligned}
&\mc{M}_{Cl}(A)=(2^n+1)^{-1}[A+\tr(A)\id],\\
&\mc{M}_{Cl}^{-1}(A)=(2^n+1)A-\tr(A)\id.
\end{aligned}
\end{equation}

Every coin has two sides.
Random Clifford unitary is challenging to realize on current quantum platforms, as it requires quite a few two-qubit gates. Current methods evenly sample and compile a Clifford to $O(n)$-depth quantum circuit with time complexity $O(n^2)$ \cite{bravyi2021hadamard,maslov2018shorter}. 
Note that the synthesis of Clifford circuits is intricately tied to the qubit connectivity of platforms. Over Linear-Nearest-Neighbour architecture, the upper bound of two-qubit circuit depth is $7n-4$, while for all-to-all architectures, it suggests $1.5n+O(\log^2(n))$ \cite{maslov2023cnot}. In the following Result parts, Sec.~\ref{sec:MCM} and \ref{sec:circuit}, we aim to simplify the circuit construction of Clifford measurement to its minimal form with the help of MUB for the all-to-all architecture.

At the end of this section, we recast the essentials of Clifford unitary and Pauli group for later use. The Pauli group for $n$-qubit quantum system is $\mathbb{P}^n=\{\pm \sqrt{\pm{1}}\otimes\{\id_2,X,Y,Z\}\}^{\otimes n}$, with $\id_2$, $X,Y,Z$ being the identity and Pauli operators for single-qubit. For convenience, we denote the quotient group without the phase as $\textbf{P}^n=\mathbb{P}^n/\pm \sqrt{\pm 1}$, and all non-identity Pauli operator as $\textbf{P}_*^n=\textbf{P}^n \setminus \id$. 
Clifford unitary is the normalizer of $\mbb{P}^n$, i.e., $U^{\dag}PU\in \mbb{P}^n, \forall P\in \mathbb{P}^n$, and a specific Clifford unitary is determined by its action on $U^{\dag}X_iU$ and $U^{\dag}Z_iU$ for $i\in[n]$, as $X_i$ and $Z_i$
can generate all Pauli operators. Hereafter we use $Z_i(X_i)$ for short to denote an $n$-qubit operator $Z_i\otimes\id_{[n]/\{i\}}$ ($X_i\otimes\id_{[n]/\{i\}}$) with the $i$-th qubit being non-idenity.

\section{Minimal Clifford Measurement from Mutually Unbiased Bases}\label{sec:MCM}
To simplify the full Clifford measurement in the original shadow estimation, we raise the following task which aims to minimize the size of the Clifford group while preserving the effective quantum channel $\mathcal{M}_{Cl}$. 

\begin{task}
\label{task:main}
 Suppose $\mc{E}$ is a subset of the full Clifford group $\mc{E}_{\mathrm{Cl}}$, the task reads
 \begin{equation}
\begin{split}
    & \textbf{min:  }  |\mc{E}| \\
& \textbf{s.t.:  } \forall \rho, 
 \frac{1}{|\mc{E}|} \sum_{U\in \mc{E}}{\mc{M}(\rho|U)}=\mc{M}_{Cl}(\rho),
\end{split}
\label{Equ: opt_task}
\end{equation}
where $\mc{M}_{Cl}(\rho)$ is the effective quantum channel shown in Eq.~\eqref{Ch&Inv}.
\end{task}

First of all, we show a lower bound of $|\mc{E}|$ by introducing a lemma that accounts for the effect of individual Clifford unitary. After the Clifford $U$ rotation, the final measurement is conducted in the computational basis, that is, the $Z$-basis. As a result, we only need to care about the action of $U$ on the $Z$-basis. Define the generators in the Heisenberg picture as 
\begin{equation}\label{eq:generator}
g_{i}={{U}^{\dag}}Z_iU,\ i\in[n].
\end{equation}
In other words, different Clifford $U$ with the same (ignoring the phase) $\bk{g_{i}}_{i=0}^{n-1}$ lead to the same measurement setting in the shadow estimation, as  shown explicitly in the following Lemma \ref{th:ClF_decom}.

\begin{lemma}\label{th:ClF_decom}
Denote the operators generated by $\bk{g_{i}}_{i=0}^{n-1}$ given in Eq.~\eqref{eq:generator} as $S_{\mb{m}}=\prod_{i=0}^{n-1} g_{i}^{{m}_{i}}$ with $\mb{m}$ an $n$-bit vector, the quantum channel $\mathcal{M}(\rho|U)$ defined in Eq.~\eqref{MUchannel} shows 
\begin{equation}\label{MUchannelS}
\mathcal{M}\left( \rho|U \right)={2}^{-n}\underset{{\mb{m}}\in {\{ 0,1\}^n }}{\mathop \sum }\tr( \rho S_{\mb{m}}){S_{\mb{m}}}.
\end{equation}
\end{lemma}
The proof is left in 
Appendix \ref{ap:Lemma1}, which transforms the summation of computational basis $\mb{b}$ to the operators $S_{\mb{m}}$. In fact, $S_{\mb{m}}$ are the stabilizers for the state $\ket{\Psi_0}=U^{\dag}\ket{\mb{0}}$ with $S_{\mb{m}}\ket{\Psi_0}=\ket{\Psi_0}, \forall \mb{m}$. $g_i$ and $S_{\mb{m}}$ are Hermitian operators by definition, thus the possible phase does not affect the measurement setting  in Eq.~\eqref{MUchannelS}. 
Hereafter, we ignore the global phase of them, that is $g_i,S_{\mb{m}}\in \textbf{P}^n$.

\begin{prop}\label{th:Mini_UB}
The cardinality of the unitary ensemble $\mc{E}$ in Task \ref{task:main} is lower bounded by $|\mc{E}| \geq 2^n+1$.
\end{prop}
The proof given in 
Appendix \ref{ap:th1}
is based on the fact that one individual Clifford unitary acquires $2^n-1$ non-identity Pauli information as shown in Lemma \ref{th:ClF_decom}, and there are totally $4^n-1$ ones. In this case, the Clifford elements in $\mc{E}$ do not share non-identity $S_{\mb{m}}$, one has $|\mc{E}|=(4^n-1)/(2^n-1)=2^n+1$ to cover all the Paulis for tomographic completeness.

\begin{definition}\label{df:MCM}
Suppose a subset of the Clifford group denoted by $\mc{E}_{min}$ reaches the lower bound in Proposition \ref{th:Mini_UB}, we call the corresponding measurement as the minimal Clifford measurement (MCM).
\end{definition}
Note that MCM only simplifies the Clifford measurement, but keeps the post-processing. As a result, the realization of MCM shadow estimation follows the same routine as the original one, except the random Clifford circuit ensemble applied before the final projective measurement.
Next, by introducing the Mutually Unbiased Bases (MUB) \cite{bandyopadhyay2002new,durt2005mutually,durt2010mutually}, we find a typical set $\mc{E}_{min}=\mc{E}_{\mathrm{MUB}}$ that saturates the lower bound. Here the Clifford unitaries in $\mc{E}_{\mathrm{MUB}}$ are written in the form of the stabilizer generators $\bk{g_i}$.

\begin{prop}\label{th:MCMbyMUB}
The MUB in Ref.~\cite{bandyopadhyay2002new} is an MCM, and the generators for all $2^n+1$ different Clifford unitaries are written by the stabilizer generators as follows. 
The $0$-th element of $\mc{E}_{\mathrm{MUB}}$ is $\bk{Z_i}_{i=0}^{n-1}$, i.e., the $Z$-basis measurement;
and the other $2^n$ elements labeled by $v=0,1,...,2^n-1$  are

\begin{equation}\label{Equ:MUB2}
\begin{aligned}
\bk{g_i= \sqrt{-1}^{\alpha_{v,i,i}} X_i\bigotimes_{j=0}^{n-1} Z_j^{\alpha_{v,i,j}}},
\end{aligned}
\end{equation}
with $\alpha_{v,i,j}=[(v \odot 2^i)M_{n}^{(0)}]_j$ being a binary value. Here $\odot$ denotes the multiplication in Galois Field $GF(2^n)$, $(\cdot)$ transforms the number to an $n$-bit string, and $M_n^{(0)}$ is an $n\times n$ binary matrix.
\end{prop}
 The specific definition of $M_n^{(0)}$ 
 is in 
 Appendix \ref{ap:MUB_Galois}, and the detailed proof of $\mc{E}_{\mathrm{MUB}}$ being an MCM is left in 
 Appendix \ref{ap:MUB_proof}. In fact, Eq.~\eqref{Equ:MUB2} offers another way to represent the MUB previously introduced in Ref.~\cite{durt2010mutually}. 

Some remarks about the quantum-design are demonstrated as follows. Note that the measurement channel in Eq.~\eqref{MUchannel}, Eq.~\eqref{Mchannel} and Eq.~\eqref{MUchannelS} are second-order functions of the quantum state, so any (projective) $2$-design ensemble returns the same channel as the full Clifford measurement in Eq.~\eqref{Ch&Inv}. It was shown that any covariant-Clifford 2-design of an $n$-qubit system reaches its minimum as an MUB \cite{zhu2015mutually}. As a result, MUB leads to the same channel as the full Clifford.  In some sense, our approach here manifests this result from the perspective of the action of individual Clifford element via Lemma \ref{th:ClF_decom}.

\section{Efficient circuit synthesis of MCM}\label{sec:circuit}
In this section, we focus on the sampling and synthesis of Clifford circuit of MCM shadow estimation, especially for the unitary ensemble $\mc{E}_{\mathrm{MUB}}$ given in Proposition \ref{th:MCMbyMUB}. The synthesis applies Tableau formula under the framework of Gottesman-Knill theorem \cite{aaronson2004improved,gottesman2002introduction}. Interestingly, as  $|\mc{E}_{\mathrm{MUB}}|\ll|\mc{E}_{\mathrm{Cl}}|$, there is possible room for further simplification on the quantum circuits. Finally, we give a unified method to sample and synthesize on the module-level, which makes the realization of MCM shadow very efficient.

Here we first briefly introduce the simplified Z-Tableau and leave more details in 
Appendix \ref{ap:Z-Tableau}. As shown in Sec.~\ref{sec:MCM}, the essential information is the stabilizer generators $\bk{g_i}$ in Eq.~\eqref{eq:generator} without considering the phase factor. In this way, we write 
\begin{equation}\label{eq:geneNoPhase}
\Tilde{g}_i=\bigotimes_{j=0}^{n-1} X_j^{\gamma_{ij}} Z_j^{\delta_{ij}},\ i\in[n],
\end{equation}
where $\gamma_{ij}$ and $\delta_{ij}$ are elements of two $n\times n$ binary matrices $C$ and $D$, and $T=[C,D]$ is called the Z-Tableau. 
The action of Clifford gates $V\Tilde{g}_iV^{\dag}$ on $\Tilde{g}_i$ is thus recorded as the transformation of the matrix $T$. As a result, our task is to find $V$ such that it takes all $\Tilde{g}_i$ back to the original $Z_i$, i.e., take $T$ to $T_0=[\mbb{O},\mbb{I}]$, with $C=\mbb{O}$ the null matrix, and $D=\mbb{I}$ the identity matrix.
We choose the basic generating Clifford gates as the single-qubit Hadamard gate $\mathrm{H}$, phase gate $\mathrm{S}$, and two-qubit controlled-Z gate $\mathrm{CZ}$, with the corresponding update rules of Tableau $T$ listed as follows. For all $i$, 
\begin{itemize}
\item $\mathrm{H}(a)$: exchange $\gamma_{i,a}$ and $\delta_{i,a}$ ;
\item $\mathrm{S}(a)$: $\delta_{i,a}:=\gamma_{i,a}+\delta_{i,a}$;
\item $\mathrm{CZ}(a,b)$: $\delta_{i,a}:=\delta_{i,a}+\gamma_{i,b},\ \delta_{i,b}:=\delta_{i,b}+\gamma_{i,a}$,
\end{itemize}
on qubit $a$, and qubit-pair $(a,b)$ respectively \cite{aaronson2004improved,van2021simple}.

In MCM, one needs evenly sample a unitary from the ensemble $\mc{E}_{\mathrm{MUB}}$ given in Proposition \ref{th:MCMbyMUB}. For elements in $\mc{E}_{\mathrm{MUB}}$, the Tableau of the $0$-th element is already $[\mathbb{O},\mathbb{I}]$, thus we do not need to synthesize; Tableaus of the remaining $2^n$ ones in Eq.~\eqref{Equ:MUB2} are all in the form of $[\mathbb{I},D_v]$ with $v$ denoting the element label, and $\alpha_{v,i,j}=[D_v]_{i,j}$. In particular, the $i$-th row of $D_v$ is the vector $(v\odot 2^i)M_n^{(0)}$ \cite{durt2010mutually}, as shown in Proposition \ref{th:MCMbyMUB}.

The circuit sampling and synthesis for all $2^n$ Tableaus in the form of $[\mathbb{I},D_v]$ is shown as follows.
First we introduce a special property for $D_v$ to simplify the procedure. A matrix $D$ is said to be a Hankel matrix if $D_{i,j}=D_{i',j'}$ for any $i+j=i'+j'$. We prove the following result in 
Appendix \ref{ap:Hankel}.
\begin{prop}\label{th:hankel}
The matrix $D_v$ from the Tableau $[\mathbb{I},D_v]$ of the element in $\mc{E}_{\mathrm{MUB}}$ given in Eq.~\eqref{Equ:MUB2} is a Hankel matrix, for all $v=0,1,\cdots,2^n-1$. 
\end{prop}
Since $D_v$ is a Hankel matrix, it can be decomposed as 
\begin{equation}\label{eq:decomD}
D_v=\sum_{k=0}^{2n-2} \beta^v_{k} \mbb{H}_k, 
\end{equation}
where $[\mbb{H}_k]_{i,j}=\delta_{i+j,k}$, i.e., it only has 1's on the $k$-th anti-diagonal. In this way, any $D_v$ can be written as a linear combination of at most $(2n-1)$ Hankel matrices $\{\mbb{H}_k\}_{k=0}^{2n-2}$, with $(2n-1)$-bit vector $\vec{\beta}^v$ characterizing the information of elements from Eq.~\eqref{Equ:MUB2}. 


\comments{
Since $D_v$ is a Hankel matrix, it can be decomposed as $D_v=\sum_{k=0}^{2n-2} \beta^v_{k} \mbb{H}_k$. Here $[\mbb{H}_k]_{i,j}=\delta_{i+j,k}$, i.e., it only has 1's on the $k$-th anti-diagonal. In this way, any $D_v$ can be written as a linear combination of at most $(2n-1)$ Hankel matrices $\{\mbb{H}_k\}_{k=0}^{2n-2}$, with $(2n-1)$-bit vector $\vec{\beta}^v$ characterizing the information of elements from Eq.~\eqref{Equ:MUB2}. 

For each $\mbb{H}_k$ with $\beta_{k}=1$ in the summation, one can use a circuit module $M_k$ in Fig.~\ref{fig:MUB_circuit} satisfying [$\mathbb{I}$, $\mbb{H}$] $\rightarrow$ $[\mbb{I}, \mbb{H}-\mbb{H}_k]$ to transform the Tableau from the initial [$\mathbb{I}, \mbb{H}$] gradually to [$\mathbb{I}, \mathbb{O}$]. After that, by applying the Hadamard gates $\mathrm{H}^{\otimes n}$, one can transform the Tableau $[\mathbb{I}, \mathbb{O}]\rightarrow[\mathbb{O}, \mathbb{I}]$ to complete the synthesis. 
}

\textit{Circuit sampling.} Suppose $v=\sum_{i=0}^{n-1} v_i 2^i$, and the binary representation shows $v=\{v_0,\cdots,v_{n-1}\}$. We randomly sample an $n$-bit string $v$, 
and actually we need \emph{not} to calculate the coefficient $\beta^v_{k}$ in Eq.~\eqref{eq:decomD} for each $D_v$ one by one, which will determine the final quantum circuit. Denote the corresponding coefficient of each basis vector $\{0,\cdots, 1_i,\cdots ,0\}$ with $i\in[n]$ as $\beta^{(i)}_{k}$, then by linearity the coefficient of $D_v$ reads
\begin{equation}\label{betasum}
\beta^v_{k}=\sum_i \beta^{(i)}_{k}v_i .
\end{equation}
which is a matrix multiplication with the binary matrix $\{\beta^{(i)}_{k}\}$ being $(2n-1)\times n$.
Such sampling only consumes $n$ random bits, compared with $O(n^2)$ for full Clifford circuits \cite{van2021simple}. We also remake that it needs  $\mathcal{O}(n^2)$ complexity to calculate Eq.~\eqref{betasum}.

\textit{Circuit synthesis.} For each $\mbb{H}_k$ with $\beta_{k}=1$ in the summation in Eq.~\eqref{eq:decomD}, one can use a circuit module $M_k$ in Fig.~\ref{fig:MUB_circuit} satisfying [$\mathbb{I}$, $\mbb{H}$] $\rightarrow$ $[\mbb{I}, \mbb{H}-\mbb{H}_k]$ to transform the Tableau from the initial [$\mathbb{I}, D_v$] gradually to [$\mathbb{I}, \mathbb{O}$]. After that, by applying the Hadamard gates $\mathrm{H}^{\otimes n}$, one can transform the Tableau $[\mathbb{I}, \mathbb{O}]\rightarrow[\mathbb{O}, \mathbb{I}]$ to complete the synthesis. 
Hence we directly obtain an three-stage form $\mathrm{-S-CZ-H-}$, and the whole synthesis process is listed in Algorithm \ref{algo:Circuit}.

\begin{algorithm}[H]
\caption{Circuit synthesis for Z-Tableau [$\mathbb{I}$,  $D_v$]}\label{algo:Circuit}
\begin{algorithmic}[1]
\Require
The qubit number $n$, and the parameter vector of the Hankel matrix $\{\beta_k\}_{k=0}^{2n-2}$.
\Ensure
The Clifford Circuit $\bf{C}$, initialized as $\bf{C}=\emptyset$.

\For{$k= 0~\text{\textbf{to}}~2n-2$} 
\If {$\beta_k=1$}
\If {$k$ is even}
\State $\bf{C} \leftarrow \bf{C} \cup {S(\frac{k}{2})}$
\EndIf
\State Initialize $p,q=0$.
\If {$k<n-1$} \State $p=0$, $q=k$.
\Else \State $p=k-n+1$, $q=n-1$.
\EndIf
\While{$p<q$}
\State $\bf{C} \leftarrow \bf{C} \cup {CZ(p,q)}$
\State{$p=p+1$, $q=q-1$}
\EndWhile
\EndIf
\EndFor
\For{$i= 0~\text{\textbf{to}}~n-1$} 
\State $\bf{C} \leftarrow \bf{C} \cup {H(i)}$
\EndFor
\end{algorithmic}
\end{algorithm}
 Fig.~\ref{fig:MUB_circuit} shows the explicit circuit construction using Algorithm \ref{algo:Circuit}.
The module $M_k$ may contain $\mathrm{S}$ and $\mathrm{CZ}$ gates on the first $k$-qubit with $k<n$.
$M_k$ and $M_{2n-2-k}$ are identical to each other but act on qubits in reverse order. There are in total $2n-1$ modules, and the circuit depth for all modules can be parallelized to $n$ by combining $M_k$ and $M_{n+k}$ to the same layer.

\begin{figure}[h]
   \centering
    \includegraphics[width=0.5\textwidth]{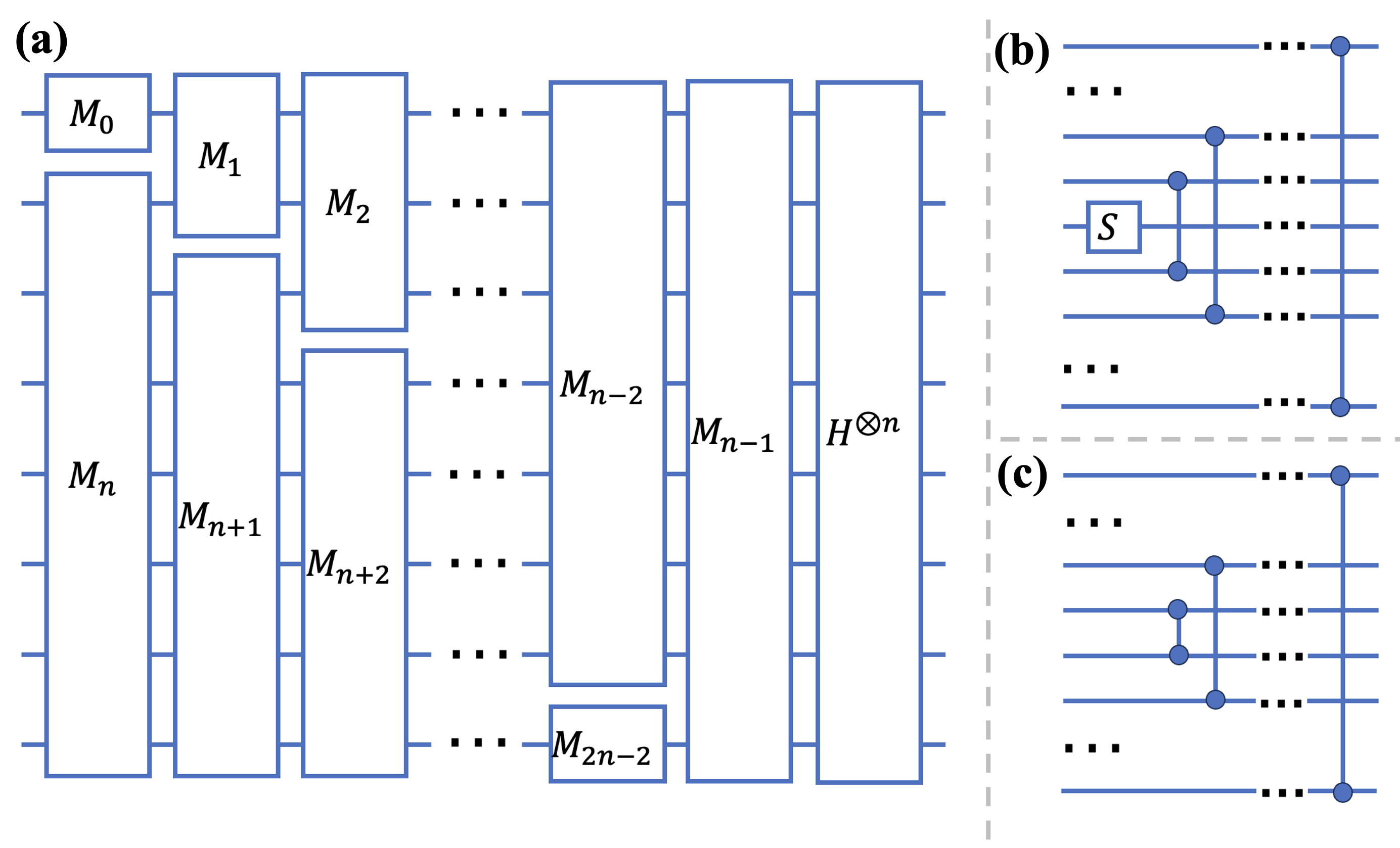}
    \caption{An illustration of the quantum circuit structure after the synthesis 
    for Tableau $[\mathbb{I}, D_v]$ with $D_v=\sum_{k=0}^{2n-2}{\beta^v_k \mbb{H}_k}$. If $\beta^v_k=1$, apply module $M_k$ as shown in (a), and the module $M_k$ and module $M_{2n-2-k}$ share the same quantum gates. If $k$ is even, the construction of the corresponding modules are shown in (b), with both $\mathrm{S}$ and $\mathrm{CZ}$ gates; 
    if k is odd, the construction of the corresponding modules are shown in (c), which only contains CZ gates.
    }
    \label{fig:MUB_circuit}
\end{figure}

We provide Fig.~\ref{fig:3q-Tab} as a visual explanation of how to eliminate the Z-Tableau and synthesize the corresponding quantum circuit.
\begin{figure}[htb]
   \centering
    \includegraphics[width=0.48\textwidth]{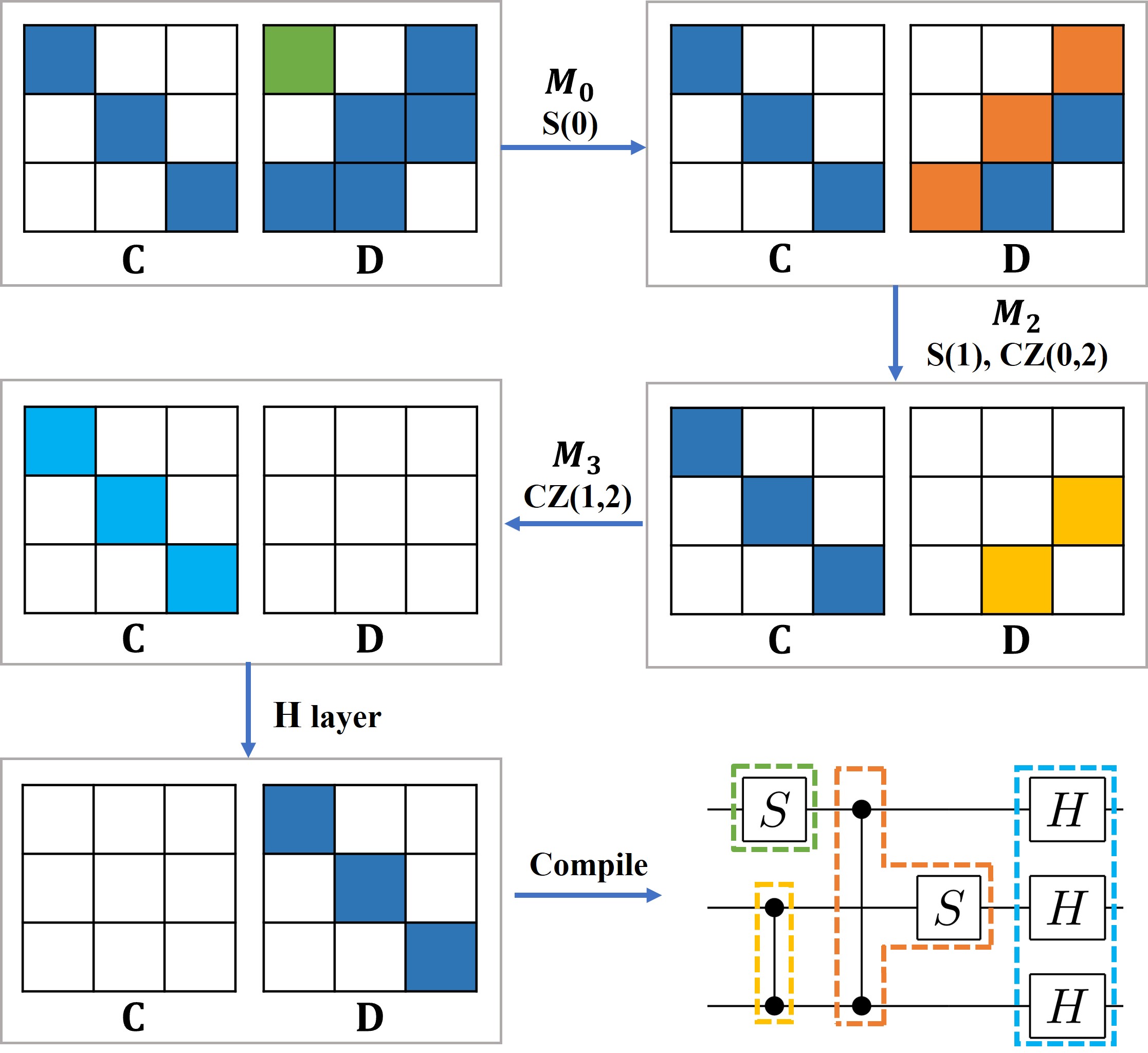}
    \caption{Illustration of the circuit synthesis in Algorithm \ref{algo:Circuit} given input Z-Tableau [$C,D$] when $n=3$. The D-Matrix of the input Z-Tableau is a Hankel matrix. To eliminate the $k$-th  anti-diagonal of the D-matrix, we make use of the $k$-th module called $M_k$. Finally a fully H-layer is utilized to transform the Z-Tableau $[\mathbb{I}, \mathbb{O}]\rightarrow[\mathbb{O}, \mathbb{I}]$.}
    \label{fig:3q-Tab}
\end{figure}
Moreover, we give an example of MUB for $n=3$, starting from the Eq~\eqref{Equ:MUB2} to the synthesis of quantum circuits, which is left in 
Appendix \ref{ap: MUBexample}.



The whole circuit can be summarized in a three-stage in the form of $\mathrm{-S-CZ-H-}$, by observing that S gates and CZ gates commute.  The total circuit depth can be further reduced, by using the result that the upper bound of the circuit depth is $\frac{n}{2}+O(\log^2(n))$ for the CZ-layer \cite{maslov2022depth}. In addition, the full Clifford group followed by projective measurements can also be effectively synthesized to a one CZ-layer form, by the canonical form of Ref.~\cite{bravyi2021hadamard} with proper mapping on the computational basis vector.
However, the modular design as shown Fig.~\ref{fig:MUB_circuit} streamlines the sampling and synthesis, and may lead to advantages in 
the calibration and realization. 
Finally, we remark that the intermediate layer of CZ gates could be implemented using one Global Mølmer–Sørensen gate \cite{figgatt2019parallel} in systems like on trapped-ion system simultaneously \cite{bravyi2022constant}.


\section{Performance analysis and off-diagonal advantage}\label{sec:performance}
With the synthesized quantum circuits for MUB at hand, one can proceed the MCM shadow estimation via Eq.~\eqref{MUchannel}, \eqref{Mchannel}, and \eqref{eq:shadowbasic}, using the same post-processing scheme as the full Clifford measurements in Eq.~\eqref{Ch&Inv}. In this section, we analyse the performance of the introduced MCM shadow in depth by investigating the variance for estimating some observable $O$. In shadow estimation, such variance shows
\begin{equation}
\mathrm{Var}_{\mc{E}}(\hat{O})=\mathrm{Var}_{\mc{E}}(\hat{O}_0)\leq \max_{\rho} \mathbb{E} \tr(O_0\widehat{\rho})^2=\|O_0 \|^2_{\mathrm{s},{\mc{E}}}.
\label{eq:Var}
\end{equation}
Here, $O_0=O-\tr(O)\id/d$ is the traceless part of $O$. Note that the variance is further bounded by the square of the shadow-norm $\|O_0\|_{\mathrm{s},{\mc{E}}}$ \cite{huang2020predicting}.  Shadow-norm is a function of the observable $O$ and the unitary ensemble $\mc{E}$, by taking the maximization over all possible input state $\rho$. Another variance quantification is called the locally-scrambled shadow-norm $\|O \|_{\mathrm{ls},{\mc{E}}}\leq \|O \|_{\mathrm{s},{\mc{E}}}$ \cite{bertoni2022shallow,bu2022classical,hu2023classical}, by considering the \emph{average} on $\rho$ from a $1$-design ensemble, or equivalently, the input state $\rho=\id/2^n$ being the maximally mixed state.

\begin{theorem}\label{thm: general variance}
For MCM shadow estimation, if one takes the unitary ensemble $\mc{E}_{\mathrm{MUB}}$ determined by some MUB, the corresponding shadow-norm and locally-scrambled shadow-norm can be upper bounded by
\begin{equation}\label{eq:VarTh}  
\begin{aligned}
&\|O_0 \|^2_{\mathrm{s},{\mc{E}_{\mathrm{MUB}}}}\leq (2^n+1) \tr(O_0^2),\\
&\|O_0 \|^2_{\mathrm{ls},\mc{E}_{\mathrm{MUB}}}= \frac{2^n+1}{2^n} \tr(O_0^2).
\end{aligned}
\end{equation}
\end{theorem}

\begin{proof}
The proof is mainly based on the 2-design property \cite{zhu2015mutually} of MUB. Here, we denote the MUB state as $\ket{\Phi_{U,\mb{b}}}=U^{\dag}\ket{\mb{b}}$, and the density matrix as $\Phi=\ket{\Phi}\bra{\Phi}$. By utilizing the inverse channel in Eq.~\eqref{Ch&Inv}, and recalling Eq.~\eqref{eq:Var}, the variance can be upper bounded by 

\begin{equation}
  \begin{aligned}\label{varMain}
&\mathbb{E} \tr(O_0\hat{\rho})^2=\mathbb{E} \tr[O_0\ \mc{M}^{-1}(\Phi_{U,\mb{b}})]^2\\
&=(2^n+1)^2\ \mathbb{E} \tr[O_0 \Phi_{U,\mb{b}}]^2\\ 
&= (2^n+1)\sum \limits_{U\in \mathcal{E}_{MUB},\mb{b}} \tr[O_0 \Phi_{U,\mb{b}}]^2 \tr(\rho \Phi_{U,\mb{b}})\\ 
&\leq (2^n+1) \tr[O_0^{\otimes 2} \sum \limits_{U\in \mathcal{E}_{MUB} ,\mb{b}}\Phi_{U,\mb{b}}^{\otimes 2}]\\
& =(2^n+1)\tr [{O}_{0}^{\otimes 2}\ (\mbb{S}+\id^{\otimes 2})]=(2^n+1)\tr(O_0^2). 
\end{aligned}
\end{equation}

Here, $\mbb{S}$ and $\id^{\otimes 2}$ are swap and identity operators on the 2-copy space.
The inequality is due to the fidelity $\tr(\rho \Phi_{U,\mb{b}})\leq 1$. And the final equality is because $\sum \Phi_{U,\mb{b}}^{\otimes 2}$ is a projective 2-design, thus the summation result is proportional to the projection to the symmetric subspace $\Pi_{sym}=2^{-n}(2^n+1)^{-1}(\mbb{S}+\id^{\otimes 2})$ on $\mc{H}_d^{\otimes 2}$. 

For the locally-scrambled shadow-norm, one just take $\rho=\id/2^n$ in Eq.~\eqref{varMain}, which contributes a $2^{-n}$ to the final result. 
\end{proof}

Some remarks on the comparison to the original full Clifford shadow are illustrated as follows.  For the worst case, i.e., the original shadow-norm is $\|O_0 \|^2_{\mathrm{s},{\mc{E}_{\mathrm{Cl}}}}\sim O(1) \tr(O_0^2)$ and our MCM shadow is exponentially worse than it, the average behaviour of both two protocols are the same by the 2-design property $\|O_0 \|^2_{\mathrm{ls},{\mc{E}_{\mathrm{MUB}}}}=\|O_0 \|^2_{\mathrm{ls},{\mc{E}_{\mathrm{Cl}}}}$.

On the other hand, the bound on $\|O_0 \|^2_{\mathrm{s},{\mc{E}_{\mathrm{MUB}}}}$ is tight in some sense, that is, the exponential scaling with the qubit-numbr $n$ is not by mathematical derivation, but the essence of MCM. Let us denote $\bk{O_0}_{U,\mb{b}}:=\tr[O_0 \Phi_{U,\mb{b}}], \bk{\rho}_{U,\mb{b}}:=\tr(\rho \Phi_{U,\mb{b}})$ for short, the third line of Eq.~\eqref{varMain} can be written as
\begin{equation}
  \begin{aligned}\label{eVarcases1}
\mathrm{Var}_{\mc{E}_{\mathrm{MUB}}}(\hat{O})\leq (2^n+1)\sum \limits_{U,\mb{b}} \bk{O_0}^2_{U,\mb{b}} \bk{\rho}_{U,\mb{b}}, 
\end{aligned}
\end{equation} 
Let us consider the fidelity estimation task, which is the main application of the original shadow estimation. In this task, $O=\Psi$ is some pure (stabilizer) state, and thus $\tr(O_0^2)$ is a constant, then we have the following result.


\begin{observation}\label{ob:order1}
Suppose there is some $\Phi_{U,\mb{b}}$, such that both $\bk{O_0}_{U,\mb{b}}=\Theta(1)$ and $\bk{\rho}_{U,\mb{b}}=\Theta(1)$, then the variance of the MUB-based shadow would scale as $\Theta(2^n)$.   
\end{observation}
The observation is direct to see, as in this case one of the terms in the summation in Eq.~\eqref{eVarcases1} $\bk{O_0}^2_{U,\mb{b}} \bk{\rho}_{U,\mb{b}}=\Theta(1)$.  That is, if $O_0$ and $\rho$ both share constant overlap with some basis state $\Phi_{U,\mb{b}}$, the final result should be exponential with $n$. 
Indeed, it is not hard to find examples to support this observation.

\begin{figure}
    \centering
    \includegraphics[scale=0.22]{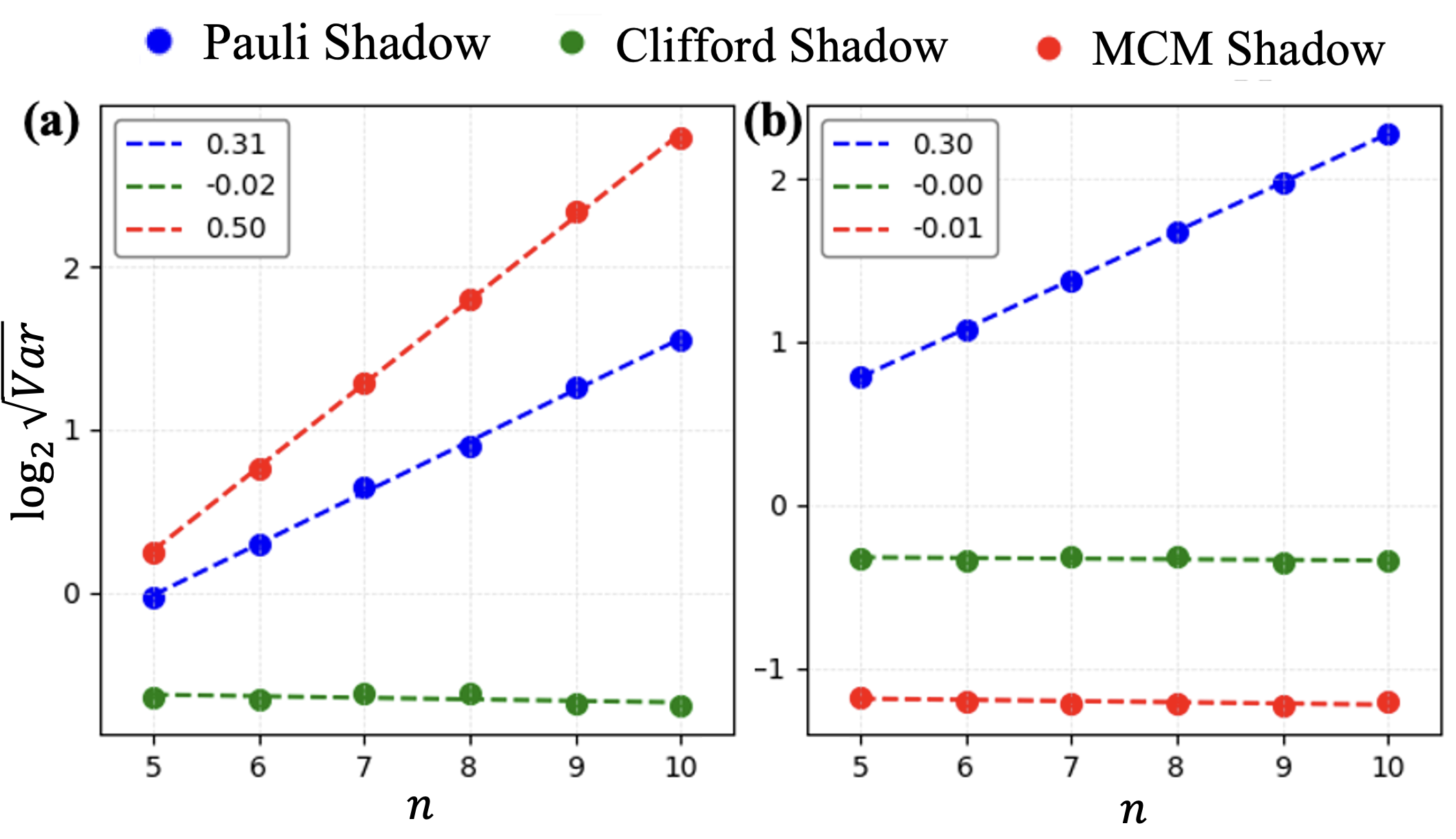}
    \caption{
    The statistical variance of shadow estimation with total number of snapshots $N=10000$ using Pauli measurement (blue), Clifford measurement (green) and Minimal Clifford measurement (red), respectively. In both (a) and (b) the state is the standard Greenberger-Horne-Zeilinger(GHZ) state $|\text{GHZ}\rangle =\frac{1}{\sqrt{2}}(|0\rangle ^{\otimes n}+|1\rangle ^{\otimes n})$. In (a)  $O=|\text{GHZ}\rangle\langle\text{GHZ}|$, and in (b) $O_F=[(|0\rangle\langle1|)^{\otimes n}+|1\rangle\langle0|)^{\otimes n}]/2$, which is the off-diagonal term of the GHZ state. The dotted lines show the fitting curves with the numerical values, and the numbers on the left-top represent
the corresponding slopes of these curves.}\label{fig:fid+off-diag}
\end{figure}

Here, in the numerical simulation we take the observable and the input state $O=\rho=\ket{\text{GHZ}}\bra{\text{GHZ}}$ as an example. In this case, there exists $\Phi_{\mbb{I},\mb{0}}$, such that $\bk{O_0}_{\mbb{I},\mb{0}},\bk{\rho}_{\mbb{I},\mb{0}}=\Theta(1)$.  That is, the state $\Phi_{\mbb{I},\mb{0}}$ is a diagonal state corresponding to the $0$-th element of MUB---the $Z$-basis with $U=\id$. As a result, the variance should be $\Theta(2^n)$, which is also manifested by the numerics in Fig.~\ref{fig:fid+off-diag} (a). 

Observation \ref{ob:order1} motivates us to relieve such worst-case exponential-scaling variance of MCM by considering the diagonal/off-diagonal part \cite{morris2019selective} of the observable $O$ separately, under a chosen basis determined by $U\in \mc{E}_{\mathrm{MUB}}$ from all possible $(2^n+1)$ elements. We demonstrate the variance for estimating only the off-diagonal part as follows.

\begin{theorem}\label{thm:off-diag}
For a given basis from one element of MUB, i.e., $\{\ket{\Phi_{U,\mb{b}}}\}$ with a fixed $U\in \mc{E}_{\mathrm{MUB}}$, consider the off-diagonal part of observable $O$, $O_F=\sum_{\mb{b} \neq \mb{b'}}{O_{\mb{b},\mb{b}'}\ket{\Phi_{U,\mb{b}}}\bra{\Phi_{U,\mb{b'}}}}$.
The variance of estimating $O_F$ with MCM-shadow is upper bounded by 
\begin{equation}\label{var:off-diag}
   \mathrm{Var}_{\mc{E}_{\mathrm{MUB}}}\left(\widehat{O}_F\right)\leq \frac{{{2}^{n}}+1}{2^{n}}{C_{l_1}(O)}^{2},
\end{equation}
where $C_{l_1}(O):=\sum_{\mb{b} \neq \mb{b'}}{|O_{\mb{b},\mb{b}'}|}$ denotes the $l_1-$norm of quantum coherence \cite{b2014q}.   
\end{theorem}
The proof is left in 
Appendix \ref{ap:off-diag}. Theorem \ref{thm:off-diag} links the scaling variance to an intrinsic quantum resource, quantum coherence \cite{streltsov2017colloquium}. 
In terms of applications, it guarantees the performance of MCM via diagonal/off-diagonal strategy when estimating the fidelity of states with a modest degree of superposition, like GHZ states and W states. As such, the variance is at most polynomial if $C_{l_1}(O)=O(\mathrm{poly(n)})$. In practice, one can estimate the diagonal part of $O$ by directly selecting the measurement basis determined by some preferred $U$, which is easy to conduct. For the challenging part, i.e., the off-diagonal part \cite{guhne2007toolbox,Hu2022Hamiltonian}, one can apply the MCM-based shadow to estimate it. Finally, by combining the diagonal and off-diagonal results, the whole estimation procedure is finalized.

In the example of estimating the fidelity of a GHZ state, the chosen basis is the Z-basis, i.e., $U=\id$. $O_F=\frac{1}{2}(|1\rangle \langle 0|^{\otimes n}+|0\rangle \langle 1|^{\otimes n})$ with $C_{l_1}(O)=1$. One can measure $O_F$ with MCM-shadow and the diagonal part with the Z-basis measurement respectively. The final variance would not scale with $n$, which is further manifested by Fig.~\ref{fig:fid+off-diag}(b). 

The off-diagonal advantage shown in Theorem \ref{thm:off-diag} and the showcase of GHZ state lies on the essence that we choose a proper basis from all MUBs as the `proper' diagonal basis, then separately measure the diagonal/off-diagonal in standard/shadow measurements. The `proper' basis is chosen such that the constant overlap $\bk{O_0}_{U,\mb{b}}=\Theta(1)$ in Observation \ref{ob:order1} can be avoided in the shadow estimation. In general, one may would like to measure a complex observable, spanning quite a few different MUBs.  This situation makes the selection of the `proper' basis frustrating, and this motivates us to consider the biased-MCM in the next section.

\section{Biased-MCM shadow estimation}\label{sec:biased-MCM}
In this section, we further propose the biased-MCM scheme as an extension and enhancement of the off-diagonal advantage discussed in the previous section. 

In the off-diagonal advantage, a reference basis is chosen deterministically in some sense, compared to the other bases which are selected randomly. Here, by leveraging prior knowledge about the observable to be predicted, we extend this kind of unbiasedness treatment to allow sampling the MUB elements with varying probabilities. Intuitively, one can increase the probability to choose some basis if it can reveal more information about the underlying observable. Note that there is biased scheme for Pauli measurements \cite{hadfield2022measurements,huang2021efficient}, but none for Clifford measurements, on account of the large cardinality of Clifford group. Here our MCM approach makes the biased scheme possible for Clifford measurements. 

Compared to the (unifrom) MCM shadow, two modifications are made as follows.
First, instead of sampled uniformly, now a unitary $U\in\mc{E}_{\mathrm{MUB}}$ is sampled with a given probability $p_U$. Second, the formula of post-processing is changed to   
\begin{equation}\label{biasedpost}
    \widehat{O_0} = \tr(O_0\frac{U^{\dag}|\mb{b}\rangle\langle \mb{b}| U}{p_U}),
\end{equation}

with the estimator of $O$ as $\hat{O}=\widehat{O_0}+2^{-n}\tr(O)$. We demonstrate the unbiasedness of $\hat{O}$ in 
Appendix~\ref{ap:unbiased estimator}.

\begin{figure}[htbp]
    \centering
    \includegraphics[scale=0.28]{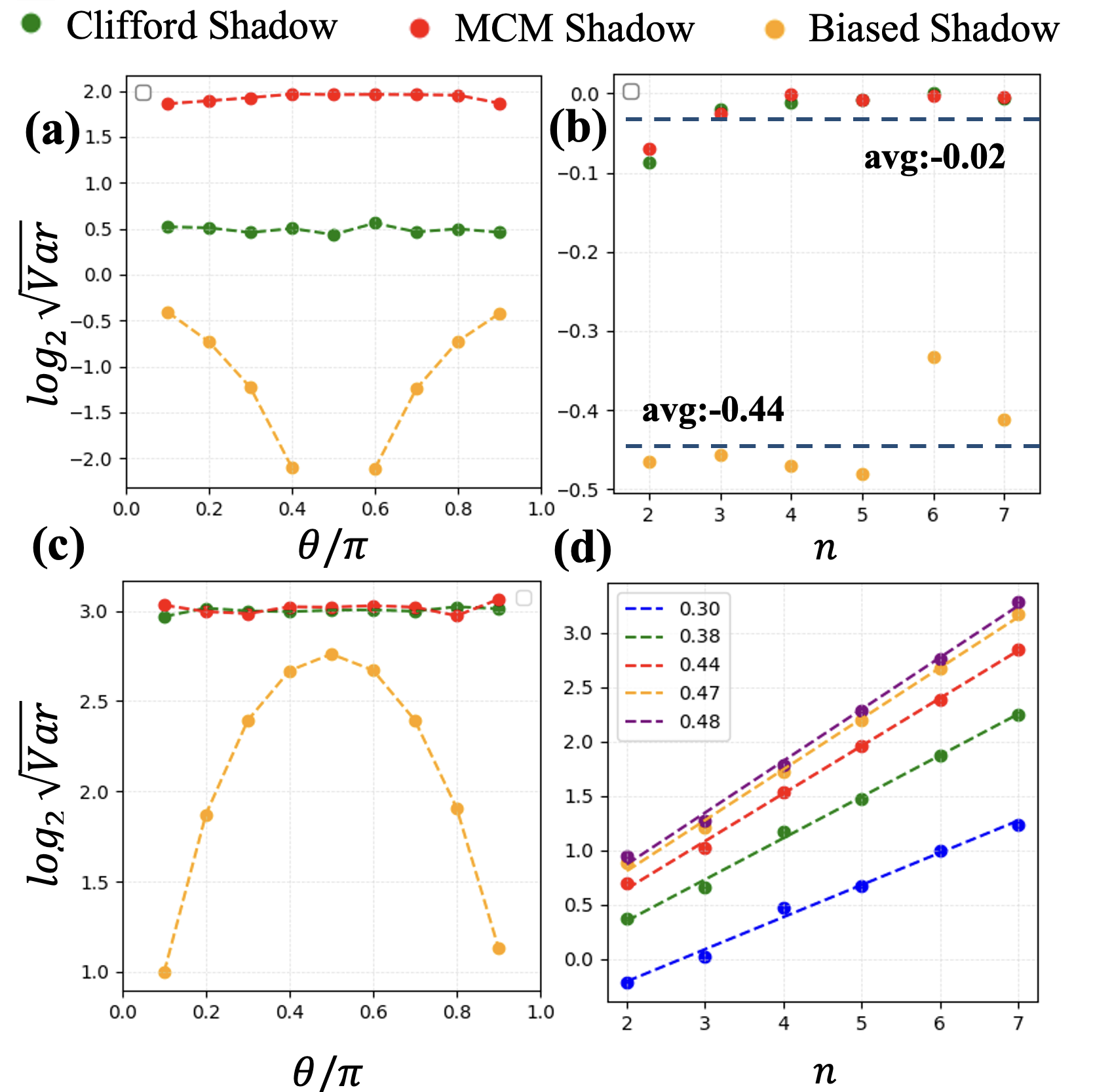}
    \caption{The statistical variance of shadow estimation with total number of snapshots $N=10000$ using Clifford measurement (green), MCM (red) and biased-MCM (orange) in three different scenarios.  In (a), the quantum state $\rho$ is parametered as the phased-GHZ states $|\text{GHZ}(\theta)\rangle = 1/\sqrt{2}[\cos(\theta/2) |0\rangle^{\otimes 6}+\sin(\theta/2) |1\rangle^{\otimes 6}]$, with $\theta\in [0.1\pi, 0.9\pi]$. And the observable $O$ set to be a $6$-qubit GHZ state. In (b), we estimate the fidelity of a set of Haar random states $\rho_i=\Phi_i$ for $i\in [1,100]$ with some Clifford stabilizer $O=\Phi_C$. The variance is the average of these 100 cases. In (c) and (d), the observable $O =[\cos(\theta/2)X+\sin(\theta/2)Z]^{\otimes n}$ and the state $\rho = |\mb{0}\rangle\langle \mb{0}|$, with $n=6$ and $\theta\in [0.1\pi, 0.9\pi]$ in (c), and in (d) the lines representing $\theta=0.1$ to $0.5$ from bottom to top. The dotted lines in (d) show the fitting curves with the numerical values, and the numbers on the left-top represent the corresponding slopes of these curves.}
    \label{Fig:Gbiased}
    \end{figure}

In general, an arbitrary observable can be decomposed into MUBs as \cite{wiesniak2011entanglement,morris2019selective}

\begin{equation}\label{MUB:Dec}
    O = -\tr(O) \mbb{I} + \sum_{U \in \mc{E}_{\mathrm{MUB}}} \sum_{\mb{b}\in{\{0,1\}}^{n}} \alpha_{U,\mb{b}} \Phi_{U,\mb{b}},
\end{equation}
where $\alpha_{U,\mb{b}} = \tr (O \Phi_{U,\mb{b}})$ with $\Phi_{U,\mb{b}}=U^{\dag}|\mb{b}\rangle \langle \mb{b}| U$. Based on this canonical decomposition, we analytically find the `optimal' probability distribution to conduct biased-MCM. Here optimal means that we can find a solution to optimize the upper bound of the variance in the estimation. In particular, 
we choose 
\begin{equation}\label{op:PU}
\begin{aligned}
p_U=\frac{B_U}{\sum_{U' \in \mc{E}_{\mathrm{MUB}} }B_{U'}},
\end{aligned}
\end{equation}
with
\begin{equation}\label{op:PU1}
\begin{aligned}
B_U:&=\max_{\mb{b}\in{\{0,1\}}^{n}} |\tr( \Phi_{U,\mb{b}}O_0)|\\
&=\max_{\mb{b}\in{\{0,1\}}^{n}}|\alpha_{U,\mb{b}}-2^{-n} \tr(O)|.
\end{aligned}
\end{equation}

Based on the biased-MCM approach according to the optimal probability for sampling unitaries in Eq.~\eqref{op:PU} and \eqref{op:PU1}, we show the following upper bound of estimation variance. 
\begin{theorem}\label{th:b-MCM}
    For an observable $O$, if one applies the biased-MCM shadow estimation using the sampling probability $p_U$ given in Eq.~\eqref{op:PU}, the variance of estimation is upper bounded by 
    \begin{equation}\label{}
    \mathrm{Var}_{\text{biased}-\mc{E}_{\mathrm{MUB}}}(\hat{O}) \leq \left(\sum_{U \in \mc{E}_{\mathrm{MUB}}} B_U \right)^2\leq \mc{D}(O_0)^2,
    \label{Equ: Var_biased}
    \end{equation} 
    where $\mc{D}(A):=2^{-n} \sum_{P\in {\mb{P}_n}} {|\tr (P A)|}$ is defined as the stabilizer norm \cite{campbell2011catalysis}.
\end{theorem}

The proof is in
Appendix \ref{ap:theorem-3}. Here we utilize a norm of non-stabilizerness (magic) \cite{campbell2011catalysis} to evaluate the estimation variance of shadow estimation, along a similar line with Theorem \ref{thm:off-diag} about quantum coherence, is also a measure of quantum resource \cite{Howard2017Resource,leone2022stabilizer}.
In particular, for $O=\rho$ being some quantum state, $\mc{D}(O_0)=D(\rho)-1/d$, where $\mc{D}(\rho)$ serves as a lower bound of robustness of magic, an important measure of non-stabilizerness that quantifies the complexity of classical simulation cost \cite{Howard2017Resource}. Therefore, any observable with a constant or polynomial magic can be estimated with a feasible sample complexity if utilizing biased-MCM.

An additional expense of the classical computation complexity for biased-MCM is to sample unitary $U$ and to compute the probability $p_U$, which unfortunately tends to be exponentially hard. However, inspired by Ref.~\cite{flammia2011direct}, we find there exists methods to tackle this exponential problem when estimated observables are Pauli strings and Clifford stabilizers, which are left in Appendix \ref{appsec:sample}. Using these methods the unitaries $U$ can be efficiently sampled and $p_U$ can be efficiently calculated. 

We further provide the variance results for the important case where the observable $O$ is a Clifford stabilizer in the following Proposition \ref{prop:zero},
by applying the interesting properties as shown in Lemma \ref{lemma:Ostab}.

\begin{prop}\label{prop:zero}
    Suppose the observable is $O=V^{\dagger}|0\rangle\langle0|V$ for any Clifford unitary $V$, one has
    \begin{equation}
    \mathrm{Var}_{\text{biased}-\mc{E}_{\mathrm{MUB}}}(\hat{O})\leq 1,
    \label{Equ: Var_biased1}
    \end{equation} 
by applying the biased-MCM with the optimized probability given in Eq.~\eqref{op:PU}. Especially, when the quantum state $\rho=O$.
\begin{equation}
    \mathrm{Var}_{\text{biased}-\mc{E}_{\mathrm{MUB}}}(\hat{O}) = 0,
    \label{Equ: Var_biased2}
    \end{equation} 
\end{prop}
\comments{
\begin{prop}\label{prop:zero}
    Suppose the observable and the unknown quantum state are both being $O=\rho=V^{\dagger}|0\rangle\langle0|V$ for any Clifford unitary $V$, one has
    \begin{equation}
    \mathrm{Var}_{\text{biased}-\mc{E}_{\mathrm{MUB}}}(\hat{O})=0,
    \label{Equ: Var_biased}
    \end{equation} 
by applying the biased-MCM with the optimized probability given in Eq.~\eqref{op:PU}.
\end{prop}
}
The proof is left in 
Appendix \ref{appsec:prop:zero}. Note that Eq.~\eqref{Equ: Var_biased1} can also be proved directly using Theorem \ref{th:b-MCM}, by the fact the stabilizer norm of a stabilizer state is just one.
 
\begin{lemma}\label{lemma:Ostab}
Suppose the observable $O= V^{\dagger}|\mb{0}\rangle \langle\mb{0}|V$, where $V$ is an arbitrary Clifford, then the element of the canonical MUB decomposition $\alpha_{U,\mb{b}}=\langle \mb{b}|UV^{\dagger}|\mb{0}\rangle|^2$ satisfies
\begin{equation}\label{Obs: sum_max}
\sum_{U\in \mc{E}_{\mathrm{MUB}}} \max_{\mb{b}} \alpha_{U,\mb{b}} =2.
\end{equation}
Furthermore, suppose the Z-Tableau of $UV^{\dag}$ is $T_{UV^{\dagger}}=[C,D]$, the distribution of $\alpha_{U,\mb{b}}$ for different $\mb{b} \in \{0,1\}^n$ are as follows: totally $2^{r_U}$ of them taking $2^{-r_U}$ and the remaining $2^n-2^{r_U}$ being 0. Here $r_U=\mathrm{rank}_{\mbb{F}_2}(C)$ on the binary field.
\end{lemma}

This lemma characterizes the projection of the stabilizer states to the MUBs, which is of independent interest. And the proof is left in 
Appendix \ref{appsec:lemma:Ostab}.

In addition, we numerically demonstrate the variance of estimating observables through different estimation protocols as shown in Fig.~\ref{Fig:Gbiased}. Fig.~\ref{Fig:Gbiased}(a) showcases the significant sampling advantage of the biased-MCM by considering the fidelity estimation of GHZ-like state. In particular, when $\rho=O=|\text{GHZ}\rangle\langle\text{GHZ}|$ at $\theta=0.5\pi$, the variance disappears. Specifically, in Fig.~\ref{Fig:Gbiased}(b), we randomly generate 100 quantum states  $\mc{E}=\{|\Phi_i\rangle\}_{i=1}^{100}$ according to the Haar measure, along with one Clifford stabilizer state $|\Phi_C\rangle$. Then we estimate the fidelity with $O=\Phi_C$ as $\tr(\Phi_i\Phi_C)=|\langle \Phi_i|\Phi_C\rangle|^2$ via different shadow protocols. The experimental result as an average from these 100 experiments highlights a constant advantage of the biased-MCM over other protocols. Moreover, the estimation variance of the observable $O =[\cos(\theta/2)X+\sin(\theta/2)Z]^{\otimes n}$ is discussed in (c) and (d). The stabilizer norm of observables shows $\mc{D}(O_0)^2 = (1+\sin \theta)^{n}$, and it grows exponentially with a base of 2 when $\theta = 0.5\pi$, which is supported by the numerics in Fig.~\ref{Fig:Gbiased} (d).

Note that the observable $O$ in Eq. \eqref{MUB:Dec} is a linear combination of several rank-1 projections $\Phi_{U,\mb{b}}$. In this framework, the role of projections $\Phi_{U,\mb{b}}$ in biased-MCM can be considered as a mimic of Pauli terms in local-biased shadow estimation \cite{hadfield2022measurements}.

\section{Conclusion and Outlook}
In this work, we propose the MCM shadow framework to simplify the original Clifford measurement to the greatest extent. By applying the canonical MUB and the Tableau formalism, we give an explicit and efficient method to synthesize the random unitary ensemble, where the budgets of both the circuit synthesizing and the experimental realization are reduced. For example, the quantum circuit can realize in the platforms with the all-to-all architectures, such as the Rydberg atoms with optical tweezers \cite{evered2023high} and trapped-ion system \cite{figgatt2019parallel}. 
The performance analysis shows the limit and advantages of MCM shadow, especially in estimating off-diagonal observables. The updated biased-MCM shadow protocol adjusts the sampling probability of the random Clifford circuit based on the target observable, which can further enhance the performance of the framework.

From the framework presented here, there are a few intriguing directions to explore in the future. First, it is very interesting to apply the MCM shadow to detect multipartite entanglement \cite{GUHNE2009detection,Friis2019Reviews}, as genuine entanglement \cite{toth2005detecting,zhou2022scheme} and more detailed structures \cite{lu2018entanglement,zhou2019detecting,zhang2023scalable} can be revealed by the fidelities to some target entangled states.  Second, the quantum circuit structure is mainly determined by the Z-Tableau according to the selection of MUBs. It is thus very worth studying how different circuit architectures \cite{Nakata2014review} and ensemble definitions \cite{zhou2022hyper,chen2023magic,park2023resource} would enhance the performance of the whole shadow protocol. Third, in the biased-MCM, we need the projection results of the target observables to all MUBs. Consequently, it is very important to develop additional adaptive or real-time method to ease this \cite{Mahler2013PRL,Ferrie2014PRL}, and find applications where the decomposition on MUB is of polynomial terms, like the Pauli operator summation for general Hamiltonian. Forth, it is possible to extend our MCM to higher-dimensional systems, where the existence of full set of MUB is not fully known \cite{durt2010mutually}. Finally, the extension to boson and fermion systems \cite{Gandhari2022CV,Becker2022cv,Zhao2021Fermionic} with a similar spirit of MCM is also intriguing.

\section{Acknowledgements}
We thank useful discussions with Zhengbo Jiang, Zhou You and Huangjun Zhu. This work is supported by National Natural Science Foundation of China(NSFC) Grant No.12205048, Innovation Program for Quantum Science and Technology 2021ZD0302000, the start-up funding of Fudan University, and the CPS-Huawei MindSpore Fellowship.

%

\onecolumngrid
\newpage

\renewcommand{\addcontentsline}{\oldacl}
\renewcommand{\tocname}{Appendix Contents}
\tableofcontents

\begin{appendix}
\bigskip

In this Appendix, 
we provide the proofs, additional discussions, and generalizations of the results presented in main text. In Sec.~\ref{MUB}, we prove Lemma \ref{th:ClF_decom} and Proposition \ref{th:Mini_UB} from main text. Moving on to Sec.~\ref{ap:MUB}, we offer supplementary details and examples on Mutually Unbiased Bases (MUB) and the circuit synthesis. Sec.~\ref{ap:Performance} discusses the variance of estimation and we also show two applications of MCM-shadow estimation. Finally, in Sec.~\ref{ap:Biased}, we provide the details and proofs of the main results of biased-MCM.


\section{Minimal Unbiased Clifford subset}\label{MUB}

\subsection{Proof of Lemma \ref{th:ClF_decom}}\label{ap:Lemma1}
\begin{proof}{\bf }

First, we define the matrix $Z_{\mb{m}}=\bigotimes_{i=0}^{n-1}{Z_i}^{m_i} $ with $\mb{m}$ an n-bit binary vector. Hence

 \begin{equation}
 \begin{split}
		Z_{\mb{m}}=\bigotimes_{i=0}^{n-1}{Z_i}^{m_i}
  &=\bigotimes_{i=1}^{n}{|0\rangle \langle 0|+(-1)^{m_i} |1\rangle \langle 1|} \\
  &=\sum_{\mb{b} \in \{0,1\}^n}{(-1)^{\mb{m}\cdot \mb{b}} |\mb{b}\rangle \langle \mb{b}|},  \\
\end{split}
		\label{Equ: Z=b}
\end{equation}\par 

where $\mb{b}$ is also an n-bit binary vector. Reversely, we have
    
 \begin{equation}
 \begin{split}
\frac{1}{2^n}\sum_{\mb{m}\in \{0,1\}^n}{(-1)^{\mb{b} \cdot \mb{m}}Z_\mb{m}}
  &=\frac{1}{2^n}\sum_{\mb{m}\in \{0,1\}^n}{(-1)^{\mb{b} \cdot \mb{m}}\sum_{\mb{m}'\in \{0,1\}^n}{(-1)^{\mb{m}'\cdot \mb{m}}|\mb{m}'\rangle \langle \mb{m}'|}}\\
  &=\frac{1}{2^n}\sum_{\mb{m},\mb{m}' \in \{0,1\}^n}{(-1)^{\mb{m}\cdot(\mb{b}+\mb{m}')}|\mb{m}'\rangle \langle \mb{m}'|}\\
  &=\frac{1}{2^n}\sum_{\mb{m}'=\mb{b},\mb{m}}{|\mb{m}'\rangle\langle \mb{m}'|}+\frac{1}{2^n}\sum_{\mb{m}'\neq \mb{b},\mb{m}}{(-1)^{\mb{m}\cdot (\mb{b}+\mb{m}')}|\mb{m}'\rangle\langle \mb{m}'|}\\
  &=|\mb{b}\rangle\langle \mb{b}|+\frac{1}{2^n}\sum_{\mb{m}'\neq \mb{b}}{|\mb{m}'\rangle\langle \mb{m}'|\sum_{\mb{m} }(-1)^{\mb{m} \cdot  (\mb{b}+\mb{m}')}}\\
  &=|\mb{b}\rangle\langle \mb{b}|.\\
\end{split}
\label{Equ: b=Z}
\end{equation}\par 

After that, we define $\Phi_{U,\mb{b}}=U^{\dag}|\mb{b} \rangle \langle \mb{b} |U$, thus we have
\begin{equation}
		\Phi_{U,\mb{b}}=\frac{1}{2^n} \sum_{\mb{m} \in \{0,1\}^n}{(-1)^{\mb{b} \cdot \mb{m}}} S_{\mb{m}}.
		\label{Phi2Sm}
	\end{equation} 
 
Now we can compute that
 \begin{equation}
 \begin{split}
		\mathcal{M}(\rho|U)
  &=\sum_{ \mb{b}\in \{0,1\}^n}{\tr\left[\rho U^{\dag} |\mb{b}\rangle \langle \mb{b}| U\right] U^{\dag} |\mb{b}\rangle \langle \mb{b}| U}\\
  &=\frac{1}{2^{2n}} \sum_{ \mb{b}\in \{0,1\}^n}{\left[\sum_{ \mb{m}\in \{0,1\}^n}{{(-1)^{\mb{b}\cdot \mb{m}}} S_{\mb{m}}}\right]\left[\sum_{ \mb{m}' \in \{0,1\}^n} (-1)^{\mb{b} \cdot \mb{m}'} \tr(\rho S_{\mb{m}'})\right]}\\
  &=\frac{1}{2^{2n}} \sum_{ \mb{b}\in \{0,1\}^n}{\left[\sum_{ \mb{m},\mb{m}'\in \{0,1\}^n}{S_{\mb{m}} \tr(\rho S_{\mb{m'}}) {(-1)^{(\mb{m}+\mb{m}')\cdot \mb{b}}} }\right]}\\
  &=\frac{1}{2^{2n}} \sum_{ \mb{m},\mb{m}'\in \{0,1\}^n}{S_{\mb{m}} \tr(\rho S_{\mb{m}'}) \sum_{ \mb{b}\in \{0,1\}^n}{ (-1)^{(\mb{m}+\mb{m}')\cdot \mb{b} }}}\\
   &=\frac{1}{2^{2n}} \sum_{ \mb{m},\mb{m}'\in \{0,1\}^n}{S_{\mb{m}} \tr(\rho S_{\mb{m}'}) \prod_{i=0}^{n-1} { \left[1+(-1)^{m_i+{m_i}'}\right] }}\\
      &=\frac{1}{2^{2n}} \sum_{\mb{m}=\mb{m}'}{S_{\mb{m}} \tr(\rho S_{\mb{m}'}) \prod_{i=0}^{n-1} { \left[1+(-1)^{m_i+{m_i}'}\right] }}=2^{-n} \sum_{ \mb{m}\in \{0,1\}^n}{S_{\mb{m}} \tr(\rho S_{\mb{m}})}.
\end{split}
		\label{label5}
\end{equation}\par 
\end{proof}

\subsection{Proof of Proposition \ref{th:Mini_UB}}\label{ap:th1}
\begin{proof}{\bf }
Firstly, the representation of $\mc{M}_{Cl}$ is converted to 
\begin{equation}
      \rho =\frac{{{2}^{n}}+1}{\left| \mathcal{E} \right|}\underset{U\in \mathcal{E}}{\mathop \sum }\,\mathcal{M}\left( \rho |U \right)-\mathbb{I}=\frac{1}{{{2}^{n}}}\underset{\sigma \in {\mb{P}^n}}{\mathop \sum }\,\tr\left( \rho \sigma  \right)\sigma.
    \label{UB_1}
\end{equation}
By introducing Lemma \ref{th:ClF_decom}, we have $\mathcal{M}\left( \rho|U \right)=\frac{1}{{{2}^{n}}}\underset{\mb{m}\ne 0}{\mathop \sum }\,{S_{\mb{m}}}\tr\left( \rho {S_{\mb{m}}} \right)+\frac{1}{{{2}^{n}}}\mathbb{I}$. Then we obtain that
\begin{equation}
\frac{{{2}^{n}}+1}{\left| \mathcal{E} \right|}\sum_{U\in \mc{E}}\sum_{\mb{m}\neq 0} \tr\left( \rho {{S}_{\mb{m}}} \right){{S}_{\mb{m}}}=\underset{\sigma \in \mb{P}_*^n~}{\mathop \sum }\,\tr\left( \rho \sigma  \right)\sigma.
    \label{UB_2}
\end{equation}
Eq.~\eqref{UB_2} must be satisfied for all quantum state $\rho$. If $|\mc{E}|<2^n+1$, there exists $\sigma \in \mb{P}_*^n$ that for all $\mb{m} \neq 0 $ and $U \in \mc{E}$, $ \sigma \neq S_{\mb{m}}$. Then let $\rho=(\sigma+\mathbb{I})/{2^n}$, we have $0=\sigma$, which is obviously a paradox! Therefore, $|\mc{E}| \geq 2^n + 1$.
\end{proof}
\section{Mutually Unbiased Bases and Circuit Synthesis}\label{ap:MUB}
\subsection{Introduction to Galois Field}\label{ap:MUB_Galois}

The example of Minimal Unbiased Clifford subset employs the construction of Mutually Unbiased Bases(MUB). Here, we first introduce the concept of Galois Field GF($2^n$) as a mathematical tool to explain the Eq. (8).

A Galois field GF$(2^n)$ has $2^n$ elements, which can be represented in the row vector form or polynomial form. 
\begin{equation}
a=(a_0,a_1,...,a_{n-1}) \quad\text{or}\quad a=\sum_{i=0}^{n-1}{a_i 2^{i}},
\label{}
\end{equation}
where $a_i$ is a binary number. So a GF$(2^n)$ element is also associated with a number from $0$ to $2^n-1$ in the polynomial form.

There are two operations in a Galois field GF$(2^n)$, the addition and the multiplication. Both of them can only be performed on two elements that belong to the same Galois field. All binary additions below are actually the XOR operations.

Galois field addition is relatively simple. Let $a$ and $b$ be two elements in the Galois field. The addition is defined as
\begin{equation}
a\oplus b=(a_0\oplus b_0,a_1\oplus b_1,...,a_{n-1}\oplus b_{n-1})=\sum_{i=0}^{n-1}{(a_i\oplus b_i) 2^{i}}.
\end{equation}

To perform Galois field multiplication correctly and ensure unique results, it is necessary to use irreducible polynomials. There are several efficient algorithms available \cite{shoup1990new,shoup1994fast} for generating these polynomials. An irreducible polynomial in binary format with $(n+1)$ bits is represented by $P_n=2^{n}+\sum_{i=1}^{n-1}{c_i 2^{i}}+1$, and the multiplication of the GF$(2^n)$ elements $a$ and $b$ is defined as
\begin{equation}
a\odot b=\sum_{k=0}^{n-1}{a_i 2^{i}} \cdot \sum_{i=0}^{n-1}{b_i 2^{i}} \text{ mod } P_n.
\label{label4}
\end{equation}
For example, an irreducible polynomial for GF$(2^2)$ is $P_2=2^{2}\oplus2\oplus1$, then $2\odot 2=2^2 \text{ mod } P_2=2\oplus 1=3$. It is easy to prove that the addition and multiplication in a Galois field satisfy the property of commutativity and associativity respectively, and they satisfy the distributive law with each other \cite{durt2010mutually}.

After introducing the fundamental operation of Galois field, we introduce the definition of the matrix $M_{n}^{(j)}$ in Eq.~\eqref{Equ:MUB2} in main text. $M_{n}^{(j)}$ is a symmetric matrix originally introduced for easy  computation of the multiplication in Galois Field. It is originally defined to calculate the multiplication of GF($2^n$) elements in the matrix form. Generally, for  GF($2^n$) elements $a,b$, we have

\begin{equation}
    [a\odot b]_j= a M_{n}^{(j)} b^{T}
\end{equation}

The calculation method for matrix $M_n^{(j)}$ is as follows. First, we define $2^k \text{ mod } P_n=(\Gamma_{k,0},\Gamma_{k,1}, ..., \Gamma_{k ,n-1})$, $k=0,1,...,2n-2$ and construct a $(2n-1)\times n$ binary matrix $\Gamma_n$: 
\begin{equation}
  \Gamma_n  =
  \begin{pmatrix}

\Gamma_{0,0} & \Gamma_{0,1} & \dots & \Gamma_{0,n-1}\\
\Gamma_{1,0} & \Gamma_{1,1} & \ddots & \vdots \\
\dots & \ddots & \ddots & \Gamma_{n-1,n-1}  \\
\Gamma_{2n-2,0} & \dots & \Gamma_{2n-2,n-2} & \Gamma_{2n-2,n-1} 
\end{pmatrix}.
		\label{label8}
\end{equation}

According to the definition of $\Gamma_n$, the first $n$ rows of the elements in $\Gamma_n$ is easily computed, and the $n$-th row of $\Gamma_n$ contains the coefficients of the irreducible polynomial $P_n$. there exists a certain recursive relationship between the elements in later rows, which can be listed as follows:
\begin{equation}
\Gamma_{i,j}=
\begin{cases}
    \delta_{i,j} & i=0,1,...,n-1 \\
    \Gamma_{n,j} & i =n\\
    (1-\delta_{j,0}) \Gamma_{i-1,j-1}+\Gamma_{n,j} \Gamma_{i-1,n-1} & i=n+1,...,2n-2,
\end{cases}
		\label{label9}
\end{equation}
where $\Gamma_{n,j}$ denotes the $j$-th cofficient of $P_n$. We can extract the elements in a certain column of $\Gamma_n$ to form a symmetric matrix $M_n^{(j)}$, where $j$ is the column number, and the elements in the matrix satisfy 
\begin{equation}
[M_n^{(j)}]_{p,q}=\Gamma_{(p+q),j}.
		\label{M_n^(j)}
\end{equation}
\subsection{Supplementary on Mutually Unbiased Bases}\label{ap:MUB_proof}

In fact, Eq.~\eqref{Equ:MUB2}
in main text is basically a rewrite of the construction of the MUB in \cite{durt2010mutually}. Here we give a brief introduction to the property of this $\mc{E}_{\mathrm{MUB}}$. Suppose the first element of $\mc{E}_{\mathrm{MUB}}$ is $\mbb{I}$. For the next $2^n$ element, we define the $v$-th element of $\mc{E}_{\mathrm{MUB}}$ is $U_v$, and $U_v$ is generated by $2^n$ Paulis. Hence we define a Pauli $S_{\mb{m},v}=U_v^{\dagger} Z_{\mb{m}} U_v\in \mb{P}^n$, where $\mathbf{m}$ is an $n$-bit binary vector, and $m_i$ represents the $i$-th bit of $\mathbf{m}$.

Notice that Pauli matrices $\{S_{\mb{m},v}\}_{\mb{m}\in\{0,1\}^n}$ form a Maximally Commuting Set(MCS) \cite{sarkar2019sets} and therefore can generate Clifford elements. Thus, we show an interesting fact as follows. 
\begin{observation}\label{ob: MCS_MUB}
$\mb{P}_*^n = \{S|S = U^{\dagger} Z_{\mb{m}} U, U\in \mc{E}_{\mathrm{MUB}}, \mb{m} \neq 0\}$.
\end{observation}

\begin{proof}
    Firstly, we define $X(a), Z(a)$  respectively as
 \begin{equation}
 \begin{split}
 &X(a)= \sum_{k=0}^{2^n-1}|k \oplus a \rangle \langle k|,\\
 &Z(a)= \sum_{k=0}^{2^n-1} (-1)^{k \odot a}|k  \rangle \langle k|.
\end{split}
\label{Equ: XZ_Galois}
\end{equation}

If $v$ and $m=\sum_{k=0}^{n-1}m_k 2^k$ are viewed as two GF($2^n$) elements, the Pauli matrices can be specified as
 \begin{equation}
 \begin{split}
S_{\mb{m},v}=\omega_{m}^{v}X(m)Z(m\odot v).
\end{split}
\label{Equ: Xz_Galois}
\end{equation}

Eq. \eqref{Equ: Xz_Galois} is actually a rewrite of Eq. (2.54) in \cite{durt2010mutually}, where the parameters $\omega_{m}^{v} \in\{\pm 1, \pm \text{i}\}$, and $S_{\mb{m},v} \in \mb{P}^n$. If $\mb{m_1}=\mb{m_2}=0$, then $S_{\mb{m_1},v_1}=S_{\mb{m_2},v_2}$, else $S_{\mb{m_1},v_1}=S_{\mb{m_2},v_2} = \mbb{I}$ if and only if $\mb{m_1}=\mb{m_2}\neq 0,v_1=v_2$. Consequently, the set $\{S_{\mb{m},v}|\mb{m}\neq 0,v=0,1,...,2^n-1\}$ have $(2^n-1)2^n$ distinct Paulis. Compared with $\mb{P}_*^n$, the set is still short of $2^n-1$ elements, which happen to be the non-identity generators of $\mbb{I}$: $\{Z_{\mb{m}}\}_{\mb{m}\neq 0}$. This is because $S_{\mb{m},v}=U_v^{\dagger} Z_{\mb{m}} U_v, U_v \neq \mbb{I}$. Hence the union of non-identity generators of MUB elements forms the set $\mb{P}_*^n$.
\end{proof}

\begin{observation}\label{ob: MUB_MCM}
 $\frac{1}{|\mc{E}_{\mathrm{MUB}}|} \sum_{U\in \mc{E}_{\mathrm{MUB}}}{\mc{M}(\rho|U)}=\mc{M}_{Cl}(\rho), \forall \rho$.
\end{observation}

 \begin{proof}
 We can use Lemma \ref{th:ClF_decom} to decompose these Cliffords.

\begin{equation}
\begin{split}
 \frac{1}{|\mc{E}_{\mathrm{MUB}}|}\sum_{U\in \mc{E}_{\mathrm{MUB}}}{\mc{M}(\rho|U)}
 &= \frac{1}{2^n+1}\sum_{U\in \mc{E}_{\mathrm{MUB}}}{\mathcal{M}\left( \rho |U \right)}\\
 &=\frac{1}{2^n(2^n+1)}\sum_{U\in \mc{E}_{\mathrm{MUB}}}\sum_{\mb{m}\ne 0}{[S_{\mb{m},v}}\tr\left( \rho {S_{\mb{m},v}} \right)+\mathbb{I}]\\
 &=\frac{1}{2^n+1}[\frac{1}{2^n}\sum_{\sigma \in \mb{P}^n}{\sigma \tr(\rho \sigma)}]+\frac{\mathbb{I}}{2^n+1}\\
 &=\frac{\rho+\mathbb{\mathbb{I}}}{2^n+1}=\mc{M}_{Cl}(\rho).\\
\end{split}
\label{Equ: MUB_proof}
\end{equation}
 \end{proof}

\subsection{Introduction to the Z-Tableau}\label{ap:Z-Tableau}
Here, we introduce the language of Tableau \cite{aaronson2004improved}.
An $n$-qubit Clifford element can be specified with four  $n \times n$  binary matrices ($\mb{\alpha},\mb{\beta},\mb{\gamma},\mb{\delta}$) and two $n$-dimensional binary vectors ($\mb{r,s}$), such that 
\begin{equation}
 {{U}^{\dag }}{{X}_{i}}U={{\left( -1 \right)}^{{{r}_{i}}}}\underset{j=0}{\overset{n-1}{\prod }}\,X_{j}^{{{\alpha }_{ij}}}Z_{j}^{{{\beta }_{ij}}}\text{   }\!\!\And\!\!\text{   }{{U}^{\dag }}{{Z}_{i}}U={{\left( -1 \right)}^{{{s}_{i}}}}\underset{j=0}{\overset{n-1}{ \prod }}\,X_{j}^{{{\gamma }_{ij}}}Z_{j}^{{{\delta }_{ij}}}.
 \label{CLF1}
\end{equation}

The matrix $X_i$, $Z_i$ are respectively defined as a Pauli matrix with $X$ applied to the $i$-th qubit, and as a Pauli matrix with $Z$ applied to the $i$-th qubit. The parameters form a $2n \times (2n+1)$ binary matrix, which is called the Tableau of a Clifford element
\begin{equation}
\left( \begin{matrix}
   \begin{matrix}
   {{[{{\alpha }_{i,j}}]}_{n\times n}}  \\
   {{[{{\gamma }_{i,j}}]}_{n\times n}}  \\
\end{matrix} & \begin{matrix}
   {{[{{\beta }_{i,j}}]}_{n\times n}}  \\
   {{[{{\delta }_{i,j}}]}_{n\times n}}  \\
\end{matrix} & \begin{matrix}
   {{[{{r}_{i,j}}]}_{n\times 1}}  \\
   {{[{{s}_{i,j}}]}_{n\times 1}}  \\
\end{matrix}  \\
\end{matrix} \right).
 \label{Tableau}
\end{equation}
According to Lemma \ref{th:ClF_decom}, 
if we substitute Clifford elements $U$ with the parameters $(\gamma,\delta)$ shown above, we still have the same $\mc{M}(\rho|U)$. Following this thread, we develop an $n \times 2n$ binary matrix called the Z-Tableau matrix
\begin{equation}
\left( \begin{matrix}
\gamma_{0,0} & \gamma_{0,1} & \dots & \gamma_{0,n-1}& \delta_{0,0} & \delta_{0,1} & \dots & \delta_{0,n-1}\\
\gamma_{1,0} & \gamma_{1,1} & \ddots & \vdots &\delta_{1,0} & \delta_{1,1} & \ddots & \vdots\\
\dots & \ddots & \ddots & \gamma_{n-2,n-1} &\vdots & \ddots & \ddots & \delta_{n-2,n-1} \\
\gamma_{n-1,0} & \dots & \gamma_{n-1,n-2} & \gamma_{n-1,n-1} &\delta_{n-1,0} & \dots & \delta_{n-1,n-2} & \delta_{n-1,n-1}
\end{matrix} \right).
\label{Z-Tableau}
\end{equation}

According to Gottesman-Knill theorem \cite{gottesman2002introduction}, Clifford circuits consist of three fundamental quantum gates: the Hadamard gate, the phase gate, and the CNOT gate, and can be efficiently simulated by classical computers in polynomial time. This means that the action of the three gates in $U$ from $U^{\dag} Z_i U$ can be expressed in terms of changes to the elements in the Z-Tableau. 

Now Clifford circuit $\mathbb{I}$ corresponds to the Z-Tableau $[\mbb{O}, \mathbb{I}]$, since $\mathbb{I}^{\dag} Z_i \mathbb{I}= \prod_{j=0}^{n-1} X_j^{\gamma_{ij}=0} Z_j^{\delta_{ij}=\delta_{i,j}}$. Moreover, it can be inferred that if a Clifford circuit $U$ can transform the Z-Tableau from $[C,D]$ to $[\mbb{O},\mathbb{I}]$, then it is a Clifford circuit that satisfies the condition $ U^{\dag} Z_i U= \prod_{j=0}^{n-1} X_j^{\gamma_{ij}} Z_j^{\delta_{ij}} $.

\subsection{Proof of Proposition \ref{th:hankel}}\label{ap:Hankel}
\begin{proof}
We use the following three equations to demonstrate the proposition,
where the first equation is determined by the property of the irreducible polynomial, and the next two come from the recurrence relation in Eq.~\eqref{label9}.

\begin{itemize}
    \item $\Gamma_{n,0}=1$,

    \item $\Gamma_{i+1,0}=\Gamma_{i,n-1}$,

    \item $\Gamma_{i+1,j}=\Gamma_{i,j-1}+\Gamma_{n,j}\Gamma_{i,n-1} (j\geq 1)$.
\end{itemize}

By the definition of Hankel matrix, the linear combination of Hankel matrices is also a Hankel matrix. Considering that a GF$(2^n)$ element $v$ can be represented as $v=\sum_{i=0}^{n-1}v_{i}2^i$, if a matrix $\mbb{D}_i$ is a Hankel matrix with its $j$-th row vector being ${(2^{i}\odot 2^{j})M_n^{(0)}}$, then $D_v$ is also a Hankel matrix. This is evident from the fact that $D_v = \sum_{i=0}^{n-1} v_i \mbb{D}_i$. The matrix $\mbb{D}_i$ is defined such that its $j$-th row vector corresponds to the $(i+j)$-th row vector of the matrix $\Gamma_n M_n^{(0)}$. Consequently, if $\Gamma_n M_n^{(0)}$ is a Hankel matrix, all $\mbb{D}_i$s are also Hankel matrices. 

Let $M_{n}=\Gamma_n M_n^{(0)}$. The first $n$ rows of $M_{n}$ are equivalent to $M_n^{(0)}$, which is a Hankel matrix. We need to prove that $M_{i,j}=M_{i-1,j+1},j=0,1,...,n-2$ for the remaining elements. When $i=n$, we have
\begin{equation}
  \begin{split}
 M_{n,j}
 &=\sum_{k=0}^{n-1}\Gamma_{n,k} M_{k,j}^{(0)}=\sum_{k=0}^{n-1}\Gamma_{n,k} \Gamma_{k+j,0}=\sum_{k=n}^{n+j-1}\Gamma_{n,k-j} \Gamma_{k,0}\\
 &= \Gamma_{n,n-j}+\sum_{k=n+1}^{n+j-1}\Gamma_{n,k-j} \Gamma_{k,0}\\
&= \Gamma_{n,n-j}+\sum_{k=n+1}^{n+j-1}\Gamma_{k,k-j}-\Gamma_{k-1,k-j-1}\\
&=\Gamma_{n+j-1,n-1}=\Gamma_{n+j,0}=M_{n-1,j+1},
  \end{split}
	\label{Hankel1}
\end{equation}
when $i\geq n+1$, we have
\begin{equation}
  \begin{split}
 M_{i,j}
 &=\sum_{k=0}^{n-1}\Gamma_{i,k} M_{k,j}^{(0)}=\sum_{k=0}^{n-1}\Gamma_{i,k} \Gamma_{k+j,0}=\sum_{k=n}^{n+j-1}\Gamma_{i,k-j} \Gamma_{k,0}\\
 &= \sum_{k=n}^{n+j-1}\Gamma_{i-1,k-j-1} \Gamma_{k,0}+\sum_{k=n}^{n+j-1}\Gamma_{n,k-j}\Gamma_{i-1,n-1} \Gamma_{k,0}\\
&= M_{i-1,j+1}+\Gamma_{n+j,0} \Gamma_{i-1,n-1}+\Gamma_{n,n-j} \Gamma_{i-1,n-1}+\sum_{k=n+1}^{n+j-1}\Gamma_{n,k-j}\Gamma_{k-1,n-1}\Gamma_{i-1,n-1}\\
&=M_{i-1,j+1}+\Gamma_{n+j,0} \Gamma_{i-1,n-1}+\Gamma_{n,n-j} \Gamma_{i-1,n-1}+\Gamma_{i-1,n-1}(\Gamma_{n+j-1,n-1}-\Gamma_{n,n-j})\\
&=M_{i-1,j+1}.
  \end{split}
		\label{Hankel2}
\end{equation}

Therefore, we demonstrate that $M_n$ is a Hankel matrix, then by definition, the matrix $D_v$ is also a Hankel. 
\end{proof}

\subsection{The $n=3$ example of MUB }\label{ap: MUBexample}

In this part, we give an example of $\mc{E}_{\mathrm{MUB}}$ in $n=3$ case, where the total number of the elements is 9. The 0-th element of $\mc{E}_{\mathrm{MUB}}$ is undoubtedly $\mathbb{I}$, and the rest 8 elements are constructed with the help of the Z-Tableau language. Suppose the Z-Tableau of the $v$-th element is $[\mathbb{I}, D_v]$, where $v=0,1,...,7$. Eq.~\eqref{Equ:MUB2}
in main text shows that the $j$-th row of matrix $D_v$ is $(v \odot 2^j) M_n^{(0)}$. Since $v=\sum_{i=0}^{2}{v_i 2^i}$, $D_v$ can be rewrote as $D_v=\sum_{i=0}^{2}{v_i \mbb{D}_{i}}$, where the $j$-th row of matrix $\mbb{D}_i$ is $(2^{j+i}) M_3^{(0)}$. Moreover, considering $(j+i)=0,1,...,4$. The binary matrix $\Gamma_n$ in Eq. (\ref{label8}) can be used to compute $\mbb{D}_i$ easily. 

Here is an example for the case of $n=3$, where $P_3=2^3\oplus 2 \oplus 1$. First, we compute $\Gamma_3$ according to Eq.~\eqref{label9}
\begin{equation}
  \Gamma_3=
  \begin{pmatrix}

1 & 0& 0\\
0 & 1& 0\\
0 & 0& 1\\
1 & 1& 0\\
0 & 1& 1
\end{pmatrix}.
		\label{MUchannel0}
\end{equation}
And $M_3^{(0)}$ can be constructed from the first row of $\Gamma_3$
\begin{equation}
  M_3^{(0)}=
  \begin{pmatrix}

1 & 0& 0\\
0 & 0& 1\\
0 & 1& 0
\end{pmatrix}.
		\label{MUchannel1}
\end{equation}
Therefore, we have the matrix 
\begin{equation}
 M_3=\Gamma_3 M_3^{(0)}=
  \begin{pmatrix}

1 & 0& 0\\
0 & 0& 1\\
0 & 1& 0\\
1 & 0& 1\\
0 & 1& 1
\end{pmatrix},
		\label{MUchannel2}
\end{equation}
and $\mbb{D}_i$ consists of the elements from the $i$-th to $(i+2)$-th row of matrix $M_3$. Finally, we can conduct that

\begin{equation}
\mbb{D}_0=
  \begin{pmatrix}

1 & 0& 0\\
0 & 0& 1\\
0 & 1& 0
\end{pmatrix} \quad
\mbb{D}_1=
  \begin{pmatrix}

0 & 0& 1\\
0 & 1& 0\\
1 & 0& 1
\end{pmatrix} \quad
\mbb{D}_2=
  \begin{pmatrix}

0 & 1& 0\\
1 & 0& 1\\
0 & 1& 1
\end{pmatrix},
		\label{MUchannel3}
\end{equation}
and every Z-Tableau matrix $D_v$ can be decomposed as  $D_v=\sum_{i=0}^{2}{v_i \mbb{D}_{i}}$ for $v = \sum_{i=0}^{2}v_i 2^i$. Notice that the choice of Z-Tableau can also be variant considering the choice of irreducible polynomials, and the transformation from Z-Tableau to the Clifford circuit can be easily implemented with a 3-stage computation of $\mathrm{-S-CZ-H-}$. The corresponding Clifford circuits are as follows.

\begin{figure}[htbp!]
    \centering
    \includegraphics[scale=0.3]{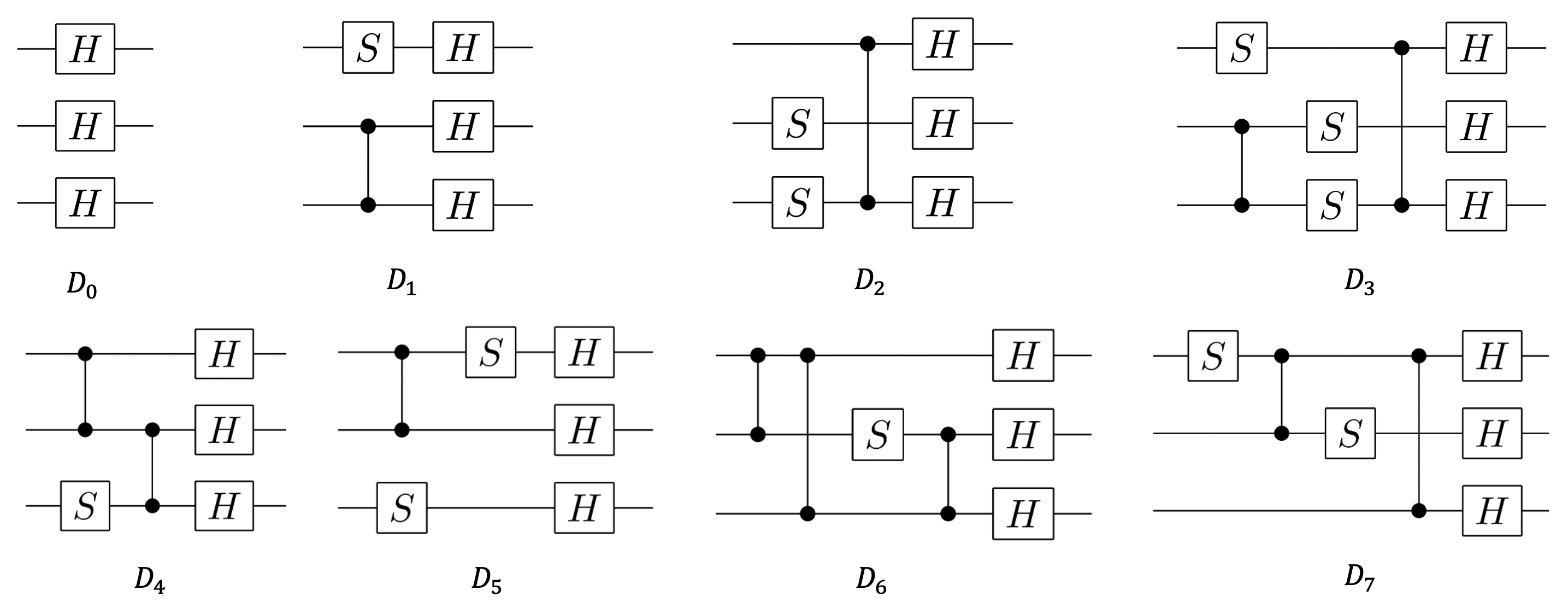}
    \caption{Illustrations of all non-identity MCM circuits for $n=3$, with $P_3=2^3\oplus 2 \oplus 1$.}
\end{figure}
    
\section{Performance analysis of MCM shadow estimation}\label{ap:Performance}

  \subsection{Variance of the estimation}  \label{ap:var}

First, we give a brief introduction to the variance of shadow estimation. In this paper, when calculating the variance, we only consider the linear physical properties and do not take into account the median-of-means method in quantum shadow tomography. Let $O$ be a fixed observable, and $\rho$ be the quantum state. After quantum shadow tomography, the classical snapshot is denoted as $\hat{\rho}$. 
And the desired property $\hat{O}=\tr(O\hat{\rho })$ should satisfy $\mathbb{E}\text{(}\hat{O})=\tr(O\rho )$.
Let $O_0 = O - \frac{\tr(O)}{2^n}\mathbb{I}$, one can calculate its variance as
\begin{equation}
\begin{split}
  & \text{Var}(\hat{O})=\mathbb{E}\left[ {{(\hat{O}-\mathbb{E}(\hat{O}))}^{2}} \right]=\mathbb{E}\left[ {{(\tr(O\hat{\rho} ))}^{2}} \right]-{{(\tr(O\rho ))}^{2}} \\ 
 & =\mathbb{E}\left[ {{(\tr({{O}_{0}}\hat{\rho}  ))}^{2}} \right]-{{(\tr({{O}_{0}}\rho ))}^{2}} \\ 
 & =\mathbb{E}\tr [O_0 \mc{M}^{-1}(\Phi_{U,\mb{b}})]^2-{{(\tr({{O}_{0}}\rho ))}^{2}}. \\ 
\end{split}
\label{var-2}
\end{equation}

Note that ${{(\tr({{O}_{0}}\rho ))}^{2}}$ is a fixed value, thus does not require detailed discussion. On the other hand, $\mathbb{E}\tr [O_0 \mc{M}^{-1}(\Phi_{U,\mb{b}})]^2$ is related to the specific measurement method, including the random measurement set $\mathcal{E}$ and the channel function $\mathcal{M}$,
\begin{equation}
\begin{split}
     &\mathbb{E}\tr [O_0 \mc{M}^{-1}(\Phi_{U,\mb{b}})]^2\\
     &={{\mathbb{E}}_{U\in \mathcal{U}}}\sum\limits_{\mb{b}\in {{\left\{ 0,1 \right\}}^{n}}}{\langle \mb{b}|U\rho {{U}^{\dag }}|\mb{b}\rangle {{\langle \mb{b}|U{{\mathcal{M}}^{-1}}({{O}_{0}}){{U}^{\dag }}|\mb{b}\rangle }^{2}}}.\\
\end{split}
    \label{var-4}
\end{equation}

Here, we briefly introduce the theoretical upper limits of the variance under two conventional methods, i.e., the Clifford measurement and the Pauli measurement. The specific derivation can be found in \cite{huang2020predicting}.
For Clifford measurement,
\begin{equation}
    \mathbb{E}\tr [O_0 \mc{M}^{-1}(\Phi_{U,\mb{b}})]^2\le 3tr({{O}^{2}}).
    \label{var-5}
\end{equation}
And for Pauli measurement, 
\begin{equation}
    \mathbb{E}\tr [O_0 \mc{M}^{-1}(\Phi_{U,\mb{b}})]^2\le {{4}^{locality(O)}}||O||_{\infty }^{2},
    \label{var-6}
\end{equation}
where $||\cdot||{}_{\infty }$ is the spectral norm (or operator norm) of a matrix, and $locality(O)$ represents the number of non-unit observables of each qubit in the observable $O$.
Meanwhile, it is noteworthy that for Clifford measurement, the average upper bound is
\begin{equation}
  \begin{split}
  & \mathbb{E}_{U \sim \mc{E}_{\text{Cl}},\rho \sim \mc{E}}\tr [O_0 \mc{M}^{-1}(\Phi_{U,\mb{b}})]^2 \\ 
 & =\frac{2^n+1}{2^n+2}{ \mathbb{E}_{\rho \sim \mc{E}}
(\tr(\rho) \tr(O_0^2) +2\tr(\rho O_0^2)) } \\ 
 & =\frac{2^n+1}{2^n+2}{
(\tr(O_0^2) +\tr(O_0^2)/{2^{n-1}}) }=\frac{2^n+1}{2^n} \tr(O_0^2).\\ 
\end{split}
\label{var-12}
\end{equation}
And for Pauli measurement, define $O=\widetilde{O} \otimes \mathbb{I}^{\otimes n-k}$ to emphasize the local observables, and $\widetilde{O}=\sum_{\textbf{p}}{\alpha_{\textbf{p}} P_{\textbf{p}}}$ where $P_\textbf{p} \in \{\mathbb{I},X,Y,Z\}^{\otimes k}$, $|\textbf{p}| $ denote the number of non-identity Paulis in $P_\textbf{p}$. The average upper bound is   

\begin{equation}
  \mathbb{E}_{U \sim \mc{E},\rho \sim \mc{E}_{\text{Pauli}}}\tr [O_0 \mc{M}^{-1}(\Phi_{U,\mb{b}})]^2=\sum_{\textbf{p}}\alpha_{\textbf{p}}^2 3^{|\textbf{p}|}.
  \label{var-13}
\end{equation}

\subsection{Off-diagonal fidelity}\label{ap:off-diag}

In this section, we discuss four specific cases for the performance of MCM and give proof for Theorem~\eqref{thm:off-diag} 
in main text. After that, we use the example of fidelity estimation to illustrate the connection between the off-diagonal observables and the original observables.
Here we provide more details on the observation of the off-diagonal fidelity. 
The four specific cases are listed as follows.

For $U\in \mc{E}_{\mathrm{MUB}}$ and $\mb{b} \in \{0,1\}^n$, $\langle\rho\rangle_{U,\mb{b}}=\tr[\rho \Phi_{U,\mb{b}}]\geq 0$, $\langle O_0\rangle_{U,\mb{b}}=\tr[O_0  \Phi_{U,\mb{b}}]\neq 0$. In the following discussion, if there exists no extra clues, we assume $\langle \rho \rangle_{U,\mb{b}},\langle O_0\rangle_{U,\mb{b}}=O(2^{-n})$.

\begin{itemize}
    \item \textbf{Case I:} There exists $U\in \mc{E}_{\mathrm{MUB}}$ and $\mb{b} \in \{0,1\}^n $ such that both $\langle \rho \rangle_{U,\mb{b}},\langle O_0\rangle_{U,\mb{b}}=\Theta(1)$. Then the upper bound is $\mathbb{E}\tr [O_0 \mc{M}^{-1}(\Phi_{U,\mb{b}})]^2=\Theta(2^n)$.\\

    \item \textbf{Case II: }There exists $ U\in \mc{E}_{\mathrm{MUB}}$ and $\mb{b} \in \{0,1\}^n$, such that $\langle \rho \rangle_{U,\mb{b}}=\Theta(1)$, but $\langle O_0 \rangle_{U,\mb{b}}=O(2^{-n})$. Then the upper bound is $\mathbb{E}\tr [O_0 \mc{M}^{-1}(\Phi_{U,\mb{b}})]^2=\Theta(1)$.\\

    \item \textbf{Case III:} There exists $U\in \mc{E}_{\mathrm{MUB}}$ and $\mb{b} \in \{0,1\}^n$, such that $\langle O_0 \rangle_{U,\mb{b}}=\Theta(1)$, but $\langle \rho \rangle_{U,\mb{b}}=O(2^{-n})$. Then the upper bound is $\mathbb{E}\tr [O_0 \mc{M}^{-1}(\Phi_{U,\mb{b}})]^2=\Theta(1)$.\\

    \item \textbf{Case IV:} For all $ U\in \mc{E}_{\mathrm{MUB}}$ and $\mb{b} \in \{0,1\}^n$,  $\langle O_0 \rangle_{U,\mb{b}}=O(2^{-n}) $ as well as $\langle \rho \rangle_{U,\mb{b}}=O(2^{-n})$. Then the upper bound is $\mathbb{E}\tr [O_0 \mc{M}^{-1}(\Phi_{U,\mb{b}})]^2=\Theta(1)$.\\

\end{itemize}

These four cases are based on the hypothesis that $\tr(O_0^2)$ is constant, which is satisfied when estimating fidelities. The proofs of the four cases are demonstrated as follows.

\textbf{Case I:}

\begin{proof}
    We set $a_1\leq\langle \rho \rangle_{U,\mb{b}} \leq b_1$,$a_2\leq\langle O_0 \rangle_{U,\mb{b}} \leq b_2$.
    \begin{equation}
  \begin{split}
  & \mathbb{E}\tr [O_0 \mc{M}^{-1}(\Phi_{U,\mb{b}})]^2 \\ 
 & =({{2}^{n}}+1){{\sum\limits_{U\in \mathcal{E}_{MUB}}{\sum\limits_{\mb{b}\in {{\{0,1\}}^{n}}}\tr [O_0 \Phi_{U,\mb{b}}]^2}}}\tr [\rho \Phi_{U,\mb{b}}] \\ 
 &\geq ({{2}^{n}}+1) a_2^2 a_1
\end{split}
\label{var-case1}
\end{equation}
Since the upper bound is generally $O(2^n)$, then it's necessarily $\Theta(2^n)$.
        \end{proof}

\textbf{Case II: }

\begin{proof}
We set $\langle O_0 \rangle_{U,\mb{b}}\leq \frac{a}{2^n}$ for all ${U,\mb{b}}$ ,
    \begin{equation}
  \begin{split}
  & \mathbb{E}\tr [O_0 \mc{M}^{-1}(\Phi_{U,\mb{b}})]^2 \\ 
 & =({{2}^{n}}+1){{\sum\limits_{U\in \mathcal{E}_{MUB}}{\sum\limits_{\mb{b}\in {{\{0,1\}}^{n}}}\tr [O_0 \Phi_{U,\mb{b}}]^2}}}\tr [\rho \Phi_{U,\mb{b}}] \\ 
 &\leq ({{2}^{n}}+1){{\sum\limits_{U\in \mathcal{E}_{MUB}}\frac{a^2}{4^n}{\sum\limits_{b\in {{\{0,1\}}^{n}}}{{}}}}}\tr [\rho \Phi_{U,\mb{b}}] \\
 &=\frac{(2^n+1)^2 a^2}{4^n} \sim O(1)
\end{split}
\label{var-case2}
\end{equation}
        \end{proof}
\textbf{     Case III: }
     
     \begin{proof}
We set $\langle \rho \rangle_{U,\mb{b}} \leq \frac{b}{2^n}$ for all ${U,\mb{b}}$,
    \begin{equation}
  \begin{split}
  & \mathbb{E}\tr [O_0 \mc{M}^{-1}(\Phi_{U,\mb{b}})]^2 \\ 
 & =({{2}^{n}}+1){{\sum\limits_{U\in \mathcal{E}_{MUB}}{\sum\limits_{b\in {{\{0,1\}}^{n}}}{\tr [O_0 \Phi_{U,\mb{b}}]^2}}}}\tr [\rho \Phi_{U,\mb{b}}] \\ 
  & \leq \frac{b(2^n+1)}{2^n}\sum_{U\in\mc{E}_{\mathrm{MUB}}}\sum_{\mb{b}\in\{0,1\}^n}\tr [O_0 \Phi_{U,\mb{b}}]^2
\end{split}
\label{var-case3}
\end{equation}
Now the equation is at the same scale with the average upper bound of Minimal Clifford measurement.
        \end{proof}

\textbf{        Case IV:    
        }
        \begin{proof}
    The conclusion is easy to proof since the scale of the upper bound in this case is certainly smaller than that of both case II and case III.
    \end{proof}

Here is the proof of Theorem 2.
\begin{proof}

For observable $O_0=O_F=\sum_{\mb{b_1}\neq\bm{b_2}} O_{\mb{b_1},\mb{b_2}}\ket{\Phi_{U,\mb{b_1}}}\bra{\Phi_{U,\mb{b_2}}}$, Eq. (\ref{var-4}) can be written as
\begin{equation}
\begin{split}
  & \mathbb{E}\tr [O_0 \mc{M}^{-1}(\Phi_{U',\mb{b}})]^2 \\ 
 &=(2^n+1)\sum_{U'\in \mathcal{E}_{MUB},\mb{b}} \tr[O_0 \Phi_{U',\mb{b}}]^2 \tr(\rho \Phi_{U',\mb{b}})\\
 &=(2^n+1)\sum_{U'\in \mathcal{E}_{MUB},\mb{b}} \tr[\sum_{\mb{b_1}\neq\bm{b_2}} O_{\mb{b}_1,\mb{b}_2}\ket{\Phi_{U,\mb{b_1}}}\bra{\Phi_{U,\mb{b_2}}}\Phi_{U',\mb{b}}]^2 \tr(\rho \Phi_{U',\mb{b}})\\
 &=(2^n+1)\sum_{U'\in \mathcal{E}_{MUB},\mb{b}}\sum_{\mb{b_1}>\mb{b_2}}\left[  O_{\mb{b}_1,\mb{b}_2}\langle\Phi_{U',\mb{b}}| \Phi_{U,\mb{b_1}}\rangle \langle\Phi_{U,\mb{b_2}}| \Phi_{U',\mb{b}}\rangle+     O_{\mb{b}_2,\mb{b}_1}\langle\Phi_{U',\mb{b}}| \Phi_{U,\mb{b_2}}\rangle  \langle\Phi_{U,\mb{b_1}}| \Phi_{U',\mb{b}}\rangle     \right]^2 \tr(\rho \Phi_{U',\mb{b}})\\
\end{split}
    \label{Equ:off-obs variance}
\end{equation}

After reviewing the definition of MUB \cite{durt2010mutually}, we have

\begin{equation}
 {
 \langle\Phi_{U',\mb{b}}| \Phi_{U,\mb{b_1}}\rangle=
 \begin{cases}
    \delta_{\mb{b},\mb{b_1}}, & U'=\mathbb{I} \\
    \frac{1}{\sqrt{2^n}}e^{i\theta_{\mb{b},\mb{b_1}}}, & U' \neq \mathbb{I}
\end{cases}
 \text{ }\!\!,\ \  \!\! \text{ }
\langle\Phi_{U',\mb{b}}| \Phi_{U,\mb{b_2}}\rangle =
 \begin{cases}
    \delta_{\mb{b},\mb{b_2}}, & U'=\mathbb{I} \\
    \frac{1}{\sqrt{2^n}}e^{i\theta_{\mb{b},\mb{b_2}}}, & U' \neq \mathbb{I}
\end{cases}
 } 
 \label{Equ: MUB_variance}
\end{equation}

where $\theta_{\mb{b},\mb{b}_1},\theta_{\mb{b},\mb{b}_2}$ are phases for certain $U,U'$. Then,

\begin{equation}
\begin{split}
& O_{\mb{b_1},{\mb{b_2}}} \braket{\Phi_{U',\mb{b}}}{\Phi_{U,\mb{b_1}}}\braket{\Phi_{U,\mb{b_2}}}{\Phi_{U',\mb{b}}} + O_{\mb{b_2},{\mb{b_1}}} \braket{\Phi_{U',\mb{b}}}{\Phi_{U,\mb{b_2}}}\braket{\Phi_{U,\mb{b_1}}}{\Phi_{U',\mb{b}}}\\
=\ & \frac{1}{2^n}\left[ O_{\mb{b_1},\mb{b_2}} e^{\text{i}(\theta_{\mb{b},\mb{b_1}}-\theta_{\mb{b},\mb{b_2}})} + O_{\mb{b_2},\mb{b_1}} e^{\text{i}(\theta_{\mb{b},\mb{b_2}}-\theta_{\mb{b},\mb{b_1}})}\right]\\
\leq\ & \frac{1}{2^n}\left( \abs{O_{\mb{b_1},\mb{b_2}}} +\abs{O_{\mb{b_2},\mb{b_1}}}\right).
\end{split}  
\end{equation}

 Now since $\sum_{\mb{b} \in \{0,1\}^n}\tr(\rho \Phi_{U',\mb{b}})=1$, we can conclude that 

\begin{equation}
\begin{split}
  & \mathbb{E}\tr [O_0 \mc{M}^{-1}(\Phi_{U',\mb{b}})]^2 \\ 
  &\leq \frac{{{2}^{n}}+1}{(2^{2n})}{{\sum\limits_{U'\neq \mathbb{I}}{\sum\limits_{\mb{b}\in {{\{0,1\}}^{n}}}{{(\sum_{\mb{b_1}\neq \mb{b_2}}{|O_{\mb{b_1},\mb{b_2}}|})^{2}}}}}}\tr [\rho \Phi_{U',\mb{b}}] \\
  &=\frac{{{2}^{n}}+1}{2^{n}}({C_{l_1}(O)})^{2}
\end{split}
    \label{Equ:off-obs variance_2}
\end{equation}
where $C_{l_1}(O)=\sum_{\mb{b_1}\neq\mb{b_2}}{|O_{\mb{b_1},\mb{b_2}}|}$ denotes the $l_1-$norm of quantum coherence \cite{b2014q}. Therefore, the scale of the theoretical upper bound is no larger than that of the norm if the norm remains a constant level.
\end{proof}

Finally, We introduce the off-diagonal fidelity as the observable in our numerical simulation.  We define the observable $O(a)=a(|0\rangle \langle 1|^{\otimes n}+|1\rangle \langle 0|^{\otimes n})+(1-a)(|0\rangle \langle 0|^{\otimes n}+|1\rangle \langle 1|^{\otimes n})$. It shows that when $a=0$, $O(a)$ turns to the off-diagonal fidelity, and when $a=0.5$, $O(a)$ turns to the fidelity of the GHZ state. Therefore, we build a connection between the fidelity and off-diagonal fidelity.

\begin{figure}[htbp!]
    \centering
    \includegraphics[scale=0.7]{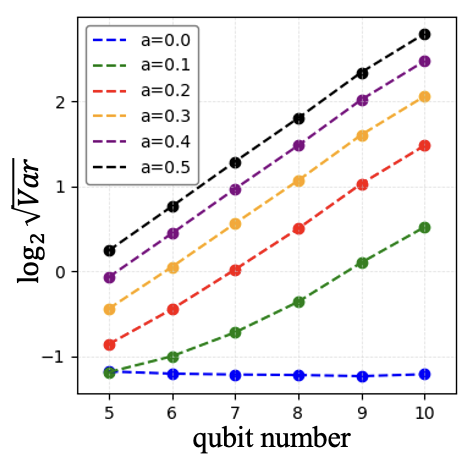}
    \caption{
    The statistical variance of shadow estimation with $N=10000$ using MCM with different parameters $a$. 
    \label{fig:fid+off-diag ex}
    }

\end{figure}

As is shown in Fig.(\ref{fig:fid+off-diag ex}), the fitting curve in $a\geq 0.2$ are approximately straight, with a slope of 0.5. However in $a=0$, the variance can be considered as a constant. The sudden transformation of the variance with the parameter $a$ demonstrate the tremendous impact of the diagonal terms  on the variance itself, especially when the qubit number increases.  This is because the diagonal terms of the observable make a exponential contribution to the variance with the qubit number. Hence the importance of dividing the off-diagonal part is demonstrated.
    \label{Fig:hybrid}

\subsection{Example 1: Estimating local observables with MCM}\label{ap:local ob}

In this example, we focus on the observable being a $k$-local operator $O_k=O(\theta)^{\otimes k}\otimes \mathbb{I}^{\otimes(n-k)}$, where $O(\theta)=\cos(\theta) Z+ \sin(\theta)X$, and the processed state $\rho=(\ketbra{0}{0})^{\otimes n}$ with the qubit number $n=8$.
We compare numerically the performance of estimating such observable using our MCM approach with the Clifford and Pauli measurement as shown in Fig. \ref{fig:local ob}. Note that each experiment repeats $N=10000$ trials.

\begin{figure}[htbp!]
    \centering
    \includegraphics[scale=0.35]{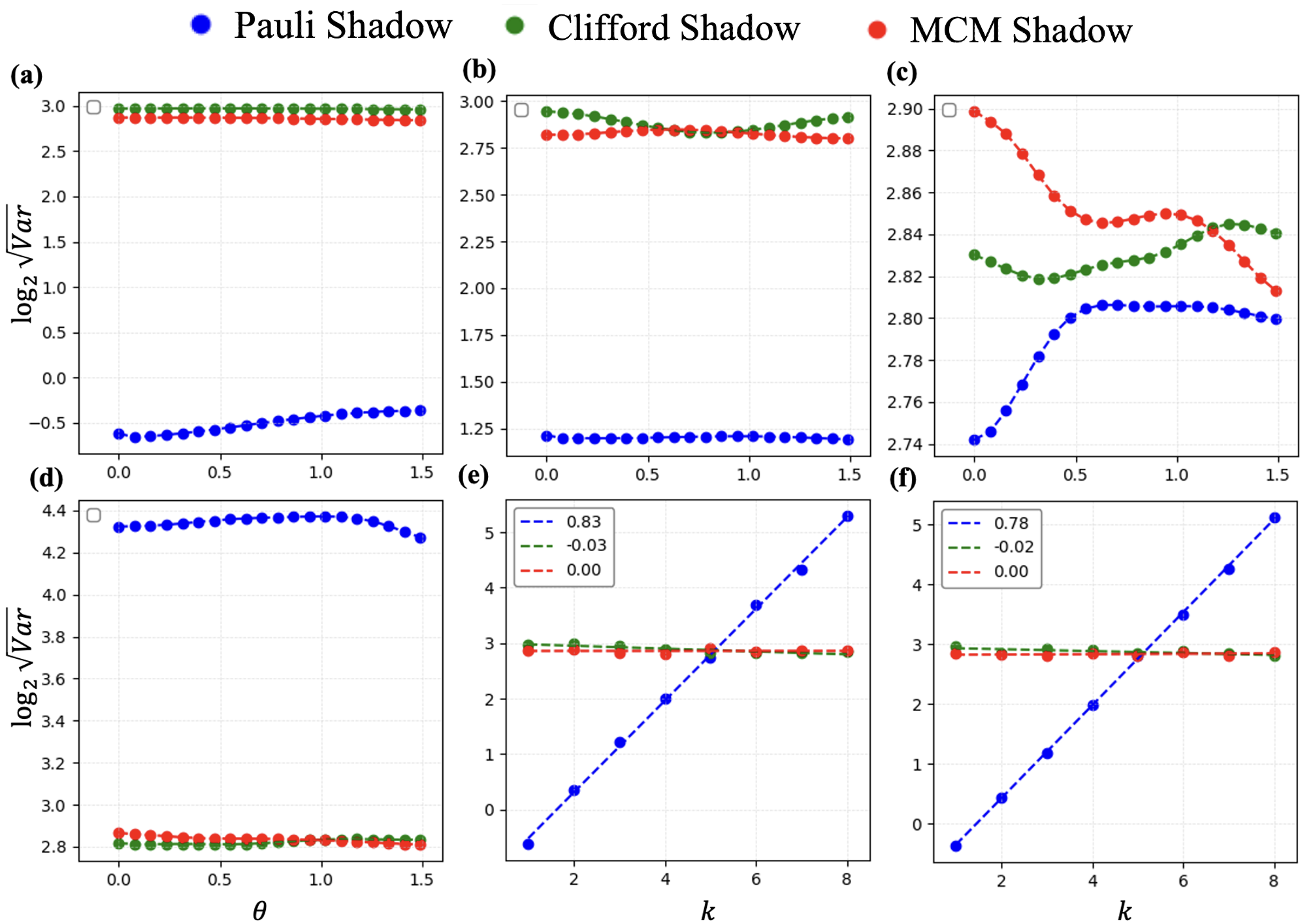}
    \caption{
   The statistical variance of estimating local observable through Pauli measurement (blue), Clifford measurement (green) and our MCM approach (red) respectively with $N=10000$ .  In (a)-(d), we investigate the dependence of $\log_2{\sqrt{Var}}$ with the phase factor $\theta$ in the observable $O_k$, with locality $k=1,3,5,7$ respectively.
   In (e) and (f), we show the dependence of the estimation variance with the locality $k$, with $\theta=0,\pi/2$ respectively. The dotted lines in (e) and (f) show the fitting curves and the corresponding slopes of these cases.
    \label{fig:local ob}
    }
    \end{figure}

We find that the performance of the MCM protocol behaves almost the same with the Clifford measurement. Specifically, the estimation variance of Pauli measurement grows exponentially with the locality number $k$, while the other two remain constant. Fig. \ref{fig:local ob}(c) shows that the intersection point of these three approaches lies in around $k=5$.
Moreover, the variance of single Pauli observables in our approach can be theoretically computed. Generally, we have $O=O_0=P$, where $P$ is a non-identity Pauli matrix. For any MUB, there exists a Clifford element $V$ and an n-bit binary vector $\mb{b}'$ such that $V^{\dag} Z_{\mb{b}'} V=P$. Thus, we have

\begin{equation}
\begin{split}
  & \mathbb{E}\tr [O_0 \mc{M}^{-1}(\Phi_{U,\mb{b}})]^2 \\ 
  &=({{2}^{n}}+1){{(\sum_{U\neq V}+\sum_{U= V}){\sum\limits_{\mb{b}\in {{\{0,1\}}^{n}}}{{(\langle \mb{b}|U P U^{\dagger}|\mb{b}\rangle)^2}}}}\langle \mb{b}|U}\rho U^{\dagger}|\mb{b}\rangle. \\
\end{split}
    \label{Equ:local_Pauli variance}
\end{equation}
After reviewing Eq. (\ref{Equ: Z=b}) and knowing that $\mb{b}' \neq 0$, we have

\begin{equation}
\begin{split}
  \langle \mb{b}|U P U^{\dagger}|\mb{b}\rangle &=\langle \mb{b}|U V^{\dag} Z_{\mb{b}'} V U^{\dag} |\mb{b}\rangle \\ 
  &=\frac{1}{2^n}\sum_{l \in\{0,1\}^n}{(-1)^{l \cdot \mb{b}'}} \langle \mb{b}| U V^{\dag}|l\rangle \langle l| V U^{\dag}|\mb{b}\rangle\\
  &=\begin{cases}
    0, & U \neq V \\
    (-1)^{\mb{b} \cdot \mb{b}'}, & U=V.
\end{cases}\\
\end{split}
    \label{Equ:local_Pauli variance}
\end{equation}
Therefore, Eq.(\ref{Equ:local_Pauli variance}) can be converted to $\mathbb{E}\tr [O_0 \mc{M}^{-1}(\Phi_{U,\mb{b}})]^2=2^n+1$, which grows exponentially with the qubit number. On the other hand, the theoretical bound of the Clifford variance is at a scale of $tr(O^2)$, which is also $O(2^n)$. The theoretical result shows that the variance of both two approaches grow with the same scale of the qubit number, and the numerical result shows that at a wider range, the performance of the Clifford measurement and the MCM are almost the same.

\subsection{ Example 2: Haar random fidelities}

\begin{figure}[htbp!]
    \centering
    \includegraphics[scale=0.3]{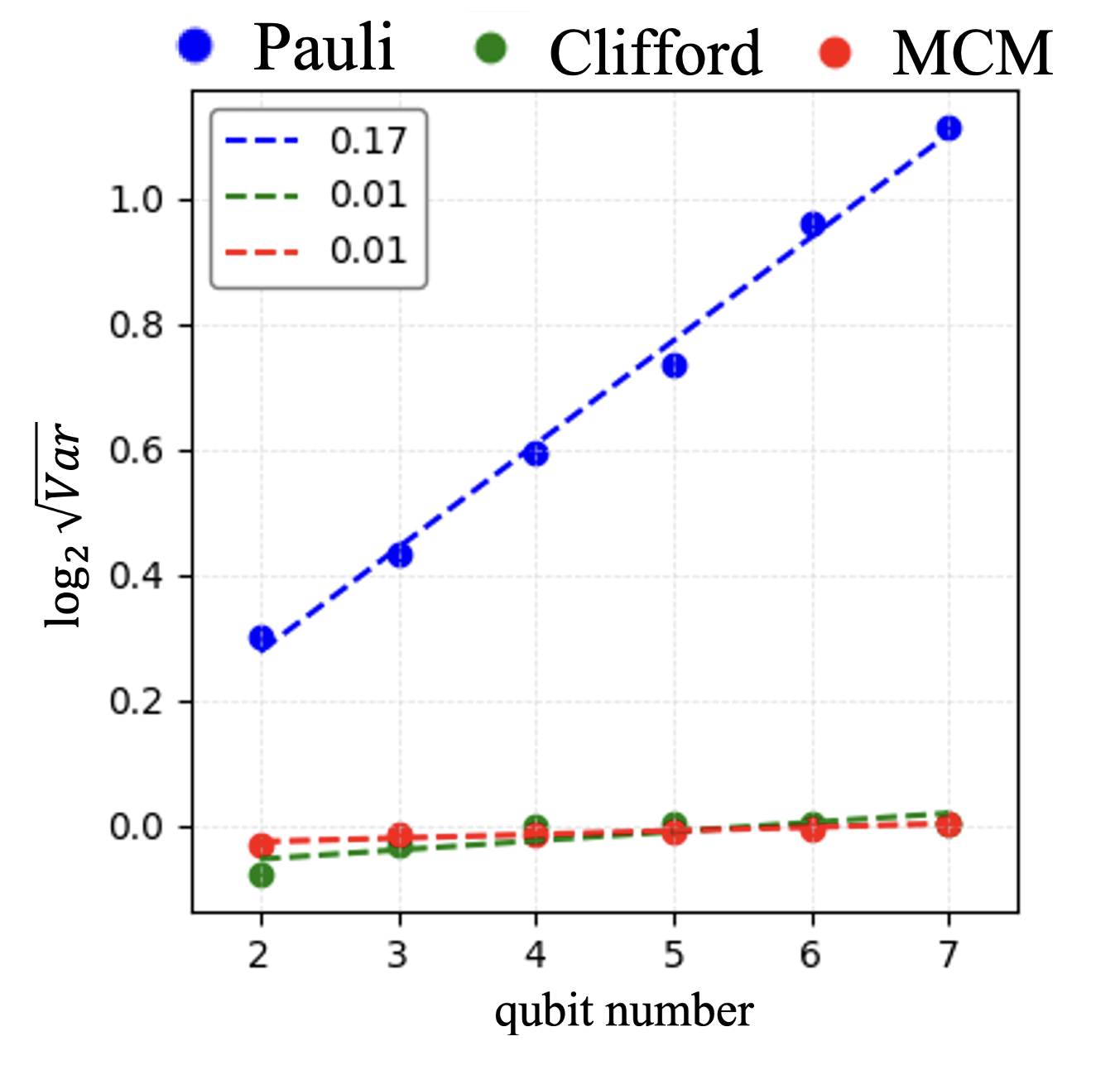}
    \caption{
    Statistical variance of shadow estimation with $N=10000$ using Pauli measurement (blue), Clifford measurement (green) and MCM protocol (red), respectively. The dotted lines show the fitting curves and the corresponding slopes of the logarithmic variance of these cases.
    \label{fig:average}
    }
\end{figure}

In this example, we utilize the random fidelities to access the average performance of different random measurements. The procedure for the simulation experiments is as follows. First we randomly generate 100 quantum states, denoted as $\mc{E}=\{|\Phi_i\rangle\}_{i=1}^{100}$ on the Haar measure, then we select the first state as the observable to estimate the fidelity $|\langle \Phi_1|\Phi_i\rangle|^2$ through shadow tomography. This process is repeated $10000$ times, and the logarithmic variance is calculated to evaluate the performance of the random measurements. Due to the property of Haar measure, the expectation  $\mathbb{E}_{\rho \sim \mc{E}}=\mathbb{I}/{2^n}$, which indicates that the variance of our approach in estimating random quantum states is on the same scale as the Clifford measurement. This is illustrated in Fig. \ref{fig:average}.

\section{Biased-MCM shadow estimation}\label{ap:Biased}

\subsection{The unbiasedness of the estimator}\label{ap:unbiased estimator}
In this section, we prove the unbiasedness of the estimator shown 
in the main text. Here we assume that all the probability $p_U \neq0$ for all possible $U\in\mc{E}_{\mathrm{MUB}}$. By definition,
\begin{equation}\label{app:unbias}
\begin{split}
    \mbb{E}_{U,\mb{b}}\widehat{O_0} 
    &=\sum_{U\in\mc{E}_{\mathrm{MUB}},\mb{b}} \frac{\bk{O_0}_{U,\mb{b}}}{p_U}\cdot p_U \cdot \bk{\rho}_{U,\mb{b}} \\
    &=\sum_{U\in\mc{E}_{\mathrm{MUB}},\mb{b}} \bk{O_0}_{U,\mb{b}}\bk{\rho}_{U,\mb{b}}\\
    &=\tr(O_0\rho)+\tr(O_0)\tr(\rho)=\tr(O_0\rho).\\
\end{split}
\end{equation}
where in the third line we apply the 2-design of MUB as that for Eq.~\eqref{varMain}
in main text. It is clear that $\hat{O}=\widehat{O_0}+2^{-n}\tr(O)$ is also unbiased by some constant shift.

\subsection{Proof of Theorem 3}\label{ap:theorem-3}
Here we calculate the variance of the estimator in this section. By definition, 
\begin{equation}
\begin{split}
    \mathrm{Var}_{\text{biased}-\mc{E}_{\mathrm{MUB}}}(\hat{O})&=\mathrm{Var}_{\text{biased}-\mc{E}_{\mathrm{MUB}}}(\widehat{O_0}) \\
    &=\mathbb{E} \widehat{O_0}^2-\bar{O_0}^2.
\end{split}  
\end{equation}
By maximizing on possible unknown state $\rho$, the first term is bounded by 
\begin{equation}
  \begin{split}
 \mathbb{E}_{U,\mb{b}} \widehat{O_0}^2&\leq\max_{\rho}\mathbb{E}_{U,\mb{b}}\widehat{O_0}^2\\
  &=\max_{\rho}\mathbb{E}_{U,\mb{b}} \frac{\bk{O_0}_{U,\mb{b}}^2}{{p_U}^2 }\\ 
 &=\max_{\rho}\sum_{U} \sum_{\mb{b}} \frac{\bk{O_0}_{U,\mb{b}}^2}{{p_U}^2 } \cdot p_U \cdot \bk{\rho}_{U,\mb{b}}\\
 &=\max_{\rho}\sum_{U} \sum_{\mb{b}} \frac{\bk{O_0}_{U,\mb{b}}^2\bk{\rho}_{U,\mb{b}}}{p_U }.
\end{split}
\label{equ:mod shadow norm}
\end{equation}
Considering the constraint that $\sum_{\mb{b}} \bk{\rho}_{U,\mb{b}}=1$, for a fixed $U$, one can further bound it using the maximal term of $\bk{O_0}_{U,\mb{b}}^2$, and finally has
\begin{equation}
    \mathrm{Var}_{\text{biased}-\mc{E}_{\mathrm{MUB}}}(\hat{O})\leq \sum_{U\in\mc{E}_{\mathrm{MUB}}}\frac{\max_{\mb{b}\in{\{0,1\}}^{n}} \bk{O_0}_{U,\mb{b}}^2}{p_U}.
    \label{equ:mod norm upper}
\end{equation}
And our optimization problem on this upper bound is as follows.
\begin{equation}
\begin{split}
    & \textbf{min ${p_U}$:}   \sum_{U\in\mc{E}_{\mathrm{MUB}}}\frac{\max_{\mb{b}\in{\{0,1\}}^{n}} \bk{O_0}_{U,\mb{b}}^2}{p_U},\\
& \textbf{s,t:  } \sum_{U \in \mc{E}_{\mathrm{MUB}}}{p_U}=1.
\end{split}
\label{Equ: opt_task}
\end{equation}
Using the Lagrange dual method, the solution to this problem is $(\sum_{U \in \mc{E}_{\mathrm{MUB}}} B_U)^2$, and the corresponding $p_U=B_U / \sum_{U' \in \mc{E} }B_{U'}$, with $B_U$ being short for $\max_{\mb{b}\in{\{0,1\}}^{n}} |\bk{O_0}_{U,\mb{b}}|$.  For the observable that is decomposed in the canonical form in Eq.\eqref{MUB:Dec} 
in main text, one has $ \bk{O_0}_{U,\mb{b}}= \alpha_{U,\mb{b}}-2^{-n}\tr(O)$.

The optimal solution can be further simplified to the following form,
\begin{equation}\label{}
\begin{split}
    \sum_{U \in \mc{E}_{\mathrm{MUB}} }B_U
    &=\sum_{U \in \mc{E}_{\mathrm{MUB}} }\max_{\mb{b}\in{\{0,1\}}^{n}} |\tr(O_0 \Phi_{U,\mb{b}})|\\
    &=\sum_{U \in \mc{E}_{\mathrm{MUB}} }\max_{\mb{b}\in{\{0,1\}}^{n}} |2^{-n}\sum_{\mb{m}\in{\{0,1\}}^{n}} (-1)^{\mb{b} \cdot \mb{m}}\tr (S_{\mb{m}}O_0)|\\
    &\leq \sum_{U \in \mc{E}_{\mathrm{MUB}} }2^{-n} \sum_{\mb{m}\in{\{0,1\}}^{n}} |\tr(S_{\mb{m}}O_0)|\\
    &=2^{-n} \sum_{P\in {\mb{P}_n^*}} {|\tr (P O_0)|} =: \mc{D}(O_0).
\end{split}
\end{equation}
Here in the second line, we express the state $\Phi_{U,\mb{b}}$ by its stabilizers using Eq.~\eqref{Phi2Sm}, and the third line is by the fact that MUBs can cover all the Paulis, and also $\tr(\id O_0)=0$. The final line we use the definition of stabilizer norm \cite{campbell2011catalysis}.
As is result, the optimal solution is upper bounded by $\mc{D}(O_0)^2$. We also remark that $\mc{D}(O_0)=\mc{D}(O)-|\tr(O)|/d$. 

Some remarks on the possibility of $p_U=0$. Suppose for some $U$, $p_U=0$ and we have $\max_{\mb{b}\in{\{0,1\}}^{n}} |\bk{O_0}_{U,\mb{b}}|=0$. That is, $\bk{O_0}_{U,\mb{b}}=0,\ \forall \mb{b}$ under this $U$. Let us denote the set with $p_U\neq 0$ as $\mc{E}'\subset \mc{E}_{\mathrm{MUB}}$. We show as follows that the possibility of $p_U=0$ does \emph{not} affect the unbiasedness proved in Appendix ~\ref{ap:unbiased estimator}. In particular, for Eq.~\eqref{app:unbias}, we can rewrite its second line by only summing the unitary in $\mc{E}'$,
\begin{equation}\label{}
\begin{split}
    \mbb{E}_{U,\mb{b}}\widehat{O_0}&=\sum_{U\in\mc{E}',\mb{b}} \bk{O_0}_{U,\mb{b}}\bk{\rho}_{U,\mb{b}}\\
    &=\sum_{U\in\mc{E}_{\mathrm{MUB}},\mb{b}} \bk{O_0}_{U,\mb{b}}\bk{\rho}_{U,\mb{b}}=\tr(\rho O_0),
\end{split}
\end{equation}
where in the second line we extend the range of summation on $U$ by the fact that when $U\in \mc{E}_{\mathrm{MUB}}/\mc{E}'$, $\bk{O_0}_{U,\mb{b}}=0, \forall \mb{b}$. 

\subsection{Proof of Lemma \ref{lemma:Ostab}}\label{appsec:lemma:Ostab}
\textbf{First, we prove the first property as follows.}
\begin{proof}
    \begin{equation}
  \begin{split}
  \alpha_{U,\mb{b}}
  &=\langle 0|V U^{\dagger}|\mb{b}\rangle \langle \mb{b}|U V^{\dagger}|0\rangle\\
  &=\frac{1}{2^n} \sum_{\mb{m}\in \{0,1\}^n} (-1)^{\mb{m} \cdot \mb{b}} \langle 0|V U^{\dagger} Z_{\mb{m}} U V^{\dagger}|0\rangle\\
\end{split}
\label{equ:sum_max_1}
\end{equation}

According to Observation \ref{ob: MCS_MUB}, we have $\bigcup_{U\in \mc{E}_{\mathrm{MUB}}}\{U^{\dagger}Z_{\mb{m}}U\}_{\mb{m}\neq 0}=\mb{P}_{*}^{n}$. And by the definition of Clifford, rotate a Clifford and we have $\bigcup_{U\in \mc{E}_{\mathrm{MUB}}}\{|VU^{\dagger}Z_{\mb{m}}UV^{\dagger}|\}_{\mb{m}\neq 0}=\mb{P}_{*}^{n}$. Also, the intersection of any two sets $\{|VU^{\dagger}Z_{\mb{m}}UV^{\dagger}|\}_{\mb{m}\neq 0}$ is $\emptyset$.

Suppose $VU^{\dagger}Z_{\mb{m}}UV^{\dagger} = (-1)^{p_{U,\mb{m}}}S_{U,\mb{m}}$, where $p_{U,\mb{m}}$ is a binary indicator of phases, and $S_{U,\mb{m}}$ is a Pauli string. Then

\begin{equation}
    \langle 0|V U^{\dagger}Z_{\mb{m}}U V^{\dagger}|0\rangle = (-1)^{p_{U,\mb{m}}}\times \mb{1} \{S_{U,\mb{m}}\triangleright Z^{\otimes n}\}.
\end{equation}

$A \triangleright Z^{\otimes n}$ \cite{huang2021efficient} means that the Pauli string $A$ consists of solely $\mbb{I}$ and $Z$. For example, $IZ,ZI,ZZ \triangleright ZZ$ and $XI \ntriangleright ZZ$. Generally $\alpha_{U,\mb{b}}$ can be rewrited as
\begin{equation}
    \alpha_{U,\mb{b}} = \frac{1}{2^n}\sum_{\mb{m}\in \{0,1\}^n} (-1)^{\mb{m}\cdot \mb{b}+p_{U,\mb{m}}}\times \mb{1} \{S_{U,\mb{m}}\triangleright Z^{\otimes n}\}\leq  \frac{1}{2^n}\sum_{\mb{m}\in \{0,1\}^n}  \mb{1} \{S_{U,\mb{m}}\triangleright Z^{\otimes n}\}.
\end{equation}

The inequality holds if and only if $\mb{m} \cdot \mb{b}+p_{U,\mb{m}}=0$ for all $\mb{m}:\{S_{U,\mb{m}}\triangleright Z^{\otimes n}\}$. Since $\mb{m}$ that satisfies these solution can form a group, where the solution for $\mb{b}$ must exist. Therefore

\begin{equation}
    \max_{\mb{b}\in \{0,1\}^n} \alpha_{U,\mb{b}}=\frac{1}{2^n}\sum_{\mb{m}\in \{0,1\}^n} \mb{1} \{S_{U,\mb{m}}\triangleright Z^{\otimes n}\}=\frac{1}{2^n}+\frac{1}{2^n}\sum_{\mb{m}\neq \mb{0}} \mb{1} \{S_{U,\mb{m}}\triangleright Z^{\otimes n}\}.
    \label{Equ: a_ub_cliff}
\end{equation}

And 

\begin{equation}
    \sum_{U\in \mc{E}_{\mathrm{MUB}}}\max_{\mb{b}\in \{0,1\}^n} \alpha_{U,\mb{b}}=\frac{2^n+1}{2^n}+\frac{1}{2^n}\sum_{S\in\mb{P}_*^n} \mb{1} \{S\triangleright Z^{\otimes n}\} = 2.
    \label{Equ:Cliff-max}
\end{equation}
\end{proof}

\textbf{Next, we prove the second property as follows.}
\begin{proof}
    In beginning of the proof, we consider the Z-Tableau of $UV^{\dagger}:$ $T_{UV^{\dagger}}=[C,D]$. We call several rowsum operations in \cite{aaronson2004improved} to perform Gaussian elimination on the matrix $C$. The result of the Gaussian elimination is the Clifford $W$, whose matrix $C$ is turned to the Row simplest form
        $\begin{pmatrix}
        C^{'}_{r_U \times n}\\
        \mbb{O} 
    \end{pmatrix}$.

    A plain conclusion is that call rowsum of Z-Tableau will only change the order of summation in Eq.~\eqref{equ:sum_max_1}, not the value of $\alpha_{U,\mb{b}}$. Therefore, 
    \begin{equation}
    \begin{split}
        \alpha_{U,\mb{b}}& = \frac{1}{2^n} \sum_{\mb{m}\in \{0,1\}^n} (-1)^{\mb{m} \cdot \mb{b}} \langle 0|W^{\dagger} Z_{\mb{m}} W|0\rangle\\
        &=\frac{1}{2^n} \sum_{\mb{m'}\in \{0,1\}^{n-r_U}} (-1)^{\mb{m'} \cdot \mb{b}} \langle 0|W^{\dagger} Z_{\mb{m'}} W|0\rangle\\
    \end{split}
    \end{equation}
The vector $\mb{m'}=[0,...,0,m_{r_U},...,m_{n-1}]$. The equation holds true because $\langle 0|W^{\dagger} Z_{\mb{m}} W|0\rangle=0$ for all $W^{\dagger} Z_{\mb{m'}} W$ that is composed of more than $\{\mbb{I},Z\}$, which is $\{W^{\dagger} Z_{\mb{m'}} W \text{ } \ntriangleright Z^{\otimes n}\}$. In this scenario, the first $r_U$ bits of $\mb{m}$ can't be all zero. Since  $W^{\dagger}Z_{\mb{m'}}W$ now consists of solely Z-Paulis, we set  $W^{\dagger}Z_{\mb{m'}}W = \prod_{j=r_U}^{n-1} (-1)^{p_j m_j} Z_j^{m_j}$. Hence

\begin{equation}
    \langle 0|W^{\dagger} Z_{\mb{m'}} W|0\rangle = \langle 0|\prod_{j=r_U}^{n-1} (-1)^{p_j m_j} Z_j^{m_j}|0\rangle.
\end{equation}
A special property for arbitrary Z-Pauli matrices $Z_1,Z_2$ is that: $\langle0|Z_1 Z_2 |0\rangle=\langle0|Z_1|0\rangle\langle0|Z_2|0\rangle$.  Hence the upper formula can be further simplified to 
\begin{equation}
\begin{split}
    \langle 0|W^{\dagger} Z_{\mb{m'}} W|0\rangle = \langle 0|\prod_{j=r_U}^{n-1} (-1)^{p_j m_j} Z_j^{m_j}|0\rangle
    = \prod_{j=r_U}^{n-1}(-1)^{m_j\cdot(p_j+\sum_{k}^{n-1}\delta_{jk})},\\
\end{split}
\end{equation}
    where $\delta_{jk}$ is the element of matrix $D$ in the Z-Tableau of $W$ (after Gaussian elimination). Now as a summarization,
    \begin{equation}
        \begin{split}
            \alpha_{U,\mb{b}} 
            &= \frac{1}{2^n} \sum_{\mb{m'}\in \{0,1\}^{n-r_U}} (-1)^{\mb{m'} \cdot \mb{b}} \langle 0|W^{\dagger} Z_{\mb{m'}} W|0\rangle\\
            &=\frac{1}{2^n} \sum_{\mb{m'}\in \{0,1\}^{n-r_U}} (-1)^{\mb{m'} \cdot \mb{b}} \prod_{j=r_U}^{n-1}(-1)^{m_j\cdot(p_j+\sum_{k}^{n-1}\delta_{jk})}\\
            &=\frac{1}{2^n} \sum_{\mb{m'}\in \{0,1\}^{n-r_U}}\prod_{j=r_U}^{n-1}(-1)^{m_j\cdot(p_j+\sum_{k}^{n-1}\delta_{jk}+b_j)}\\
            &=\frac{1}{2^n}\prod_{j=r_U}^{n-1}1+(-1)^{b_j+p_j+\sum_{k}^{n-1}\delta_{jk}}\\
            &=\begin{cases}
                2^{-r_U}, b_j+p_j+\sum_{k}^{n-1}\delta_{jk}=0 \text{  for all } j,\\
                0, \text{  otherwise}.\\
            \end{cases}
        \end{split}
        \label{Equ:rU}
    \end{equation}
    Therefore, for any $\mb{b} \in \{0,1\}^n$, $\alpha_{U,\mb{b}}$ has $2^{r_U}$ in $2^{-r_U}$ and the remaining $2^n-2^{r_U}$ in 0. 
\end{proof}
\subsection{Proof of Proposition \ref{prop:zero}}\label{appsec:prop:zero}
\begin{proof}

   Suppose for single shot, the unitary rotation is $U \in \mc{E}_{\mathrm{MUB}}$ and the measurement outcome is $|\mb{b}\rangle$. Before calculating the variance, we discuss the property of the probability $p_U$. According to the definition, $p_U=B_U/{\sum_{U'\in \mc{E}_{\mathrm{MUB}}}}B_{U'}$. By introducing Eq. \eqref{Equ:Cliff-max}, we conclude that 
\begin{equation}
    {\sum_{U'\in \mc{E}_{\mathrm{MUB}}}}B_{U'}={\sum_{U'\in \mc{E}_{\mathrm{MUB}}}} (\max_{\mb{b}}\alpha_{U,\mb{b}}-2^{-n})=1-2^{-n}.
    \label{Equ: Bu_sum}
\end{equation}

By applying Eq.~\eqref{Equ: Var_biased} 
in main text, we calculate the variance

 \begin{equation}
    \mathrm{Var}_{\text{biased}-\mc{E}_{\mathrm{MUB}}}(\hat{O}) \leq  \left(\sum_{U' \in \mc{E}_{\mathrm{MUB}}} B_{U'}\right)^2\leq 1.
    \label{Equ: Var_biased_Clifford}
    \end{equation}

A specific example is that for $\rho=O$ are Clifford stabilizers, the variance is zero. Here we prove it. 

Also according to Eq. \eqref{Equ:rU}, $\max_{\mb{b}}\alpha_{U,\mb{b}}=2^{-r_U}$. Thus,
\begin{equation}
    \sum_{U\in \mc{E}_{\mathrm{MUB}}} 2^{-r_U}=2.
    \label{Equ: rU_sum}
\end{equation}
In this task, the result of each snapshots for $U$ and $\mb{b}$ is defined as
\begin{equation}
    \hat{O} =\frac{\bk{O_0}_{U,\mb{b}}}{p_U}+\frac{\tr(O)}{2^n}.
\end{equation}
Hence the estimation is of no variance. Since $O=\rho $ is a stabilizer state, the result of $\alpha_{U,\mb{b}}$ has two options: $2^{-r_U}$ and $0$, which also equals to the probability of the result $\mb{b}$ being measured. This means that the only possible result of $\alpha_{U,\mb{b}}$ is $2^{-r_U}$.  As a result, 
\begin{equation}
    \begin{split}
        \text{Var}[\hat{O}]
        &=\mbb{E_{U,\mb{b}}}(\hat{O}-1)^2\\
        &=\sum_{U\in \mc{E}_{\mathrm{MUB}}}p_U \sum_{\mb{b} \in\{0,1\}^n} \bk{\rho}_{U,\mb{b}}[\frac{\bk{O_0}_{U,\mb{b}}}{p_U}+2^{-n}-1]^2\\
        &=\sum_{U\in \mc{E}_{\mathrm{MUB}}}p_U \sum_{\mb{b} \in\{0,1\}^n} \alpha_{U,\mb{b}}[\frac{\alpha_{U,\mb{b}}-2^{-n}}{B_U/\sum_{U'}B_{U'}}+2^{-n}-1]^2\\
        &=\sum_{U\in \mc{E}_{\mathrm{MUB}}}p_U \sum_{\mb{b} \in\{0,1\}^n} \alpha_{U,\mb{b}}[\frac{\alpha_{U,\mb{b}}-2^{-n}}{|\max_\mb{b} \alpha_{U,\mb{b}}-2^{-n}|}(1-2^{-n})+2^{-n}-1]^2\\
        &=\sum_{U\in \mc{E}_{\mathrm{MUB}}}p_U \sum_{\mb{b} \in\{0,1\}^n} \alpha_{U,\mb{b}}[ \frac{\alpha_{U,\mb{b}}-2^{-n}}{|\max_{\mb{b}}{\alpha_{U,\mb{b}}}-2^{-n}|}(1-2^{-n})+2^{-n}-1]^2\\
        &=\sum_{U\in \mc{E}_{\mathrm{MUB}}}p_U 2^{r_U} 2^{-r_U}[ \frac{2^{-r_U}-2^{-n}}{2^{-r_U}-2^{-n}}(1-2^{-n})+2^{-n}-1]^2\\
        &=0.
    \end{split}
\end{equation}

\end{proof}

\subsection{Efficient sampling of $p_U$ for Pauli strings and Clifford stabilizers}\label{appsec:sample}

The main content of this subsection is to introduce methods to efficiently sample unitary $U\in \mc{E}_{MUB}$ with a probability $p_U$ defined in Eq.~\eqref{op:PU} and \eqref{op:PU1} in the main text for both Pauli strings and Clifford stabilizers.

Initially, we give an algorithm to identify the unique $U\in \mc{E}_{MUB}$ for a given Pauli operator $P\in \mb{P}^*_n$, such that there exists $\mb{m}\in\{0,1\}^n$ satisfying $U^{\dag}Z_\mb{m}U=P$. And we denote this relation as $P \blacktriangleright U$. 

\begin{algorithm}[H]
\caption{Find $U\in \mc{E}_{MUB}$ for Pauli operator $P$ with $P \blacktriangleright U$.}\label{algo:generatorU}
\begin{algorithmic}[1]
\Require
The qubit number $n$,  the non-identity Pauli $P\in \mb{P}^*_n$, and the ensemble $\mc{E}_{MUB}$.
\Ensure
MUB element $U \in \mc{E}_{MUB}$.
\If {$P\triangleright Z^{\otimes n}$} $U=I$
\Else { Decompose $P$ as $P=\bigotimes_{i=0}^{n} i^{a_i b_i} X^{a_i} Z^{b_i}$.}
\State Solve the linear equations: 
\begin{equation}
    \sum_k(\sum_j \beta_{i+j}^{(k)}a_j)v_k = b_i \text{ for all } i
    \label{lin}
\end{equation}
to obtain the parameter $v$ so that the Z-Tableau of MUB element is $[\mathbb{I},D_v]$.
\EndIf
\end{algorithmic}
\end{algorithm}

Eq. \eqref{lin} is indeed a explicit form of $D_v \mb{a} = \mb{b}$ via Eq. \eqref{betasum}, where $\mb{a}$ and $\mb{b}$ are vectors of parameters $\{a_i\}$ and $\{b_i\}$. To solve a system of linear equations, the algorithm needs $O(n^3)$ complexity.

After introducing Algorithm \ref{algo:generatorU}, we illustrate how to efficiently sample Pauli strings and Clifford stabilizers.

\textit{Pauli string.} For a Pauli string $O = \sum_{l} \alpha_{l} P_l$, where $P \in \mb{P}^*_n$ for convenience, the length of Pauli strings $l$ tends to be a constant or polynomial to $n$. Initially, the probability term $B_U$ is computed:
\begin{equation}
        B_U =  \max_{\mb{b}}|\tr(O\Phi_{U,\mb{b}})| = \max_{\mb{b}}|\sum_{\mb{m}}(-1)^{\mb{b}\cdot\mb{m}}\alpha_{\mb{m}}|\leq \sum_{P_l\blacktriangleright U}{|\alpha_l|}.
\end{equation}

Then the probability of sampling using the upper bound of $B_U$ is calculated as follows.
\begin{equation}
    p_U = \frac{\sum_{P_l\blacktriangleright U}{|\alpha_l|}}{\sum_{l}{|\alpha_l|}}.
    \label{op:Pauli_PU}
\end{equation}
The computational cost for computing $B_U$ might be exponentially large if utilizing the original formula (Eq.\eqref{op:PU}), hence the rule of efficiently sampling would be violated. In order to avoid this dilema, we utilize Eq.\eqref{op:Pauli_PU} to calculate the probability of sampling unitaries for Pauli observables. This action will undoubtedly amplify the estimation variance since the protocol is no longer optimal. But fortunately, we find that scaling variance of biased-MCM remains for estimating Pauli observables.

\begin{equation}
    \mathrm{Var}_{\text{biased}-\mc{E}_{\mathrm{MUB}}}(\hat{O})\leq\sum_{U\in\mc{E}_{\mathrm{MUB}}}\frac{\max_{\mb{b}\in{\{0,1\}}^{n}} \bk{O_0}_{U,\mb{b}}^2}{p_U} \leq (\sum_l |\alpha_l|)^2=\mc{D}(O)^2
\end{equation}
As illustrated above, the scaling variance is equal to the result of Theorem \ref{th:b-MCM}. The sampling process shows as follows. First, search all MUB elements $U$ so that there exists $P_l\blacktriangleright U$ individually using algorithm \ref{algo:generatorU}. Second, compute the probability $p_U$ of these elements using Eq.\eqref{op:Pauli_PU}, and using the probability for sampling. In conclusion, it takes $O(n^3 l)$ time to compute the probability of sampling, which can be reused for each sample process.

\textit{Clifford stabilizer.} For a Clifford stabilizer $O = V^{\dagger}|\mb{0}\rangle \langle\mb{0}|V$ with $V$ a Clifford unitary, the sampling probability shows

\begin{equation}
    \begin{split}
        p_U
        &=\frac{B_U}{\sum_{U' \in \mc{E}_{\mathrm{MUB}} }B_{U'}}\\
        &=\frac{\max_{\mb{b}\in \{0,1\}^n} \alpha_{U,\mb{b}}-2^{-n}}{1-2^{-n}}\\
        &=\frac{1}{2^n-1} \sum_{\mb{m}\neq 0}\mb{1}\{ V U^{\dagger}Z_{\mb{m}} U V^{\dagger}\triangleright Z^{\otimes n}\}\\
        &=\frac{1}{2^n-1} \sum_{\mb{m}\neq 0}\mb{1}\{ U^{\dagger}Z_{\mb{m}} U \blacktriangleright V\}
    \end{split}
\end{equation}
via the definition of 
 $p_U$ and Eq.\eqref{Equ: a_ub_cliff}. Therefore, the sampling protocol is as follows. First, sample a non-identity generator $P_V$ of Clifford $V$ with equal probability. Second, find the sole MUB element $U$ so that $P_V \blacktriangleright U$ using algorithm \ref{algo:generatorU}. Third, compute the probability 
 \begin{equation}
     p_U = \frac{2^{-r_U}-2^{-n}}{1-2^{-n}}
 \end{equation} derived from Eq.\eqref{Equ:rU}, where $r_U=\mathrm{rank}_{\mbb{F}_2}(C)$ on the binary field, and the Z-Tableau of $UV^{\dag}$ is $T_{UV^{\dagger}}=[C,D]$. In conclusion, the classical computation complexity of each sampling process is $O(n^3)$.

\end{appendix}

\end{document}